\documentclass[12pt,english]{article}
\usepackage{lmodern}
\usepackage[T1]{fontenc}
\usepackage[utf8]{inputenc}
\usepackage[a4paper,margin=3.2cm]{geometry}
\usepackage{color}
\usepackage{babel}
\usepackage{float}
\usepackage{mathtools}
\usepackage{amsmath}
\usepackage{amssymb}
\usepackage[authoryear]{natbib}
\usepackage[unicode=true,
 bookmarks=false,
 breaklinks=false,pdfborder={0 0 1},
 colorlinks=true]
 {hyperref}
\hypersetup{
 urlcolor=refcolor,citecolor=refcolor}

\makeatletter

\usepackage{amsfonts}
\usepackage{amsthm}
\usepackage{tablefootnote}
\usepackage{xcolor}
\usepackage{url}
\usepackage{tabularx}
\usepackage{graphicx}
\usepackage{caption}
\usepackage{tikz}
\usepackage{pgfplots}
\pgfplotsset{compat=newest}
\usepgfplotslibrary{fillbetween}
\DeclareMathOperator \Esp{E}
\DeclareMathOperator \p{P}

\definecolor{refcolor}{rgb}{0, 0, 0.5}

\newtheorem{proposition}{Proposition}\newtheorem{lemma}{Lemma}\theoremstyle{definition}
\newtheorem{assumption}{Assumption}

\makeatother

\date{{\large January 25, 2022}}

\begin{document}
\title{Setting Interim Deadlines to Persuade\thanks{Thanks to Pavel Kocourek, Yiman Sun, Jan Zápal, Inés Moreno de Barreda, Jeff Ely, Jean Tirole, Péter
Esö, Margaret Meyer, Colin Stewart, Francesc Dilmé, Ansgar Walther, Ole Jann, Maxim Ivanov, Egor Starkov, Ludvig Sinander, Maxim Bakhtin, Rastislav Rehák,
 Vladimír Novák, Arseniy Samsonov, and the audiences at OLIGO 2022, SING17, CMiD2022, EEA-ESEM Congress 2022, and EWMES 2022 for helpful comments.}}
\author{Maxim Senkov \thanks{Email: msenkov@cerge-ei.cz. CERGE-EI, a joint workplace of Center for Economic Research and Graduate
Education, Charles University and the Economics Institute of the Czech
Academy of Sciences, Politických vězňů 7, P.O. Box 882, 111 21 Prague
1, Czech Republic.} \thanks{The author acknowledges financial support from the Lumina Quaeruntur fellowship (LQ300852101) of the Czech Academy of Sciences and from Czech Science Foundation (GACR) under grant no. 20-27630S.}} 
\maketitle
\begin{abstract}
{A principal funds a multistage project and retains the right to cut the funding if it stagnates at some point. An agent wants to convince the principal to fund the project as long as possible, and can design the flow of information about the progress of the project in order to persuade the principal. If the project is sufficiently promising ex ante, then the agent commits to providing only the good news that the project is accomplished. If the project is not promising enough ex ante, the agent persuades the principal to start the funding by committing to provide not only good news but also the bad news that a project milestone has not been reached by an interim deadline. I demonstrate that the outlined structure of optimal information disclosure holds irrespective of the agent’s profit share, benefit from the flow of funding, and the common discount rate.}\\
 {\footnotesize{}{}{}{}{}{} }\textbf{\footnotesize{}{}{}{}{}{}Keywords:}{\footnotesize{}{}{}{}{}{}
dynamic information design, informational incentives, interim deadline, multistage project.}\\
 {\footnotesize{}{}{}{}{}{} }\textbf{\footnotesize{}{}{}{}{}{}JEL
Classification Numbers:}{\footnotesize{}{}{}{}{}{} D82, D83,
G24, 031. }{\footnotesize\par}

\end{abstract}

\newpage
\section{Introduction}

The development of any innovation requires investment of both time and capital, while the outcome of this investment is inherently stochastic.
Usually, the investor, being the principal, retains the option to stop funding the innovative project if at some point it proves unsuccessful.
It is widely documented that the agent running the project tends to prefer the principal to postpone the stopping of the funding to enjoy either the extra funds or an additional chance to turn her research idea into a success story.\footnote{Agency conflict in which the agent prefers the principal to postpone abandoning the project that the agent is working on is studied in \cite{admati1994robust,gompers1995,bergemann1998venture,bergemann2005financing,cornelli2003stage}.} In such an agent-principal relationship, the agent's technological expertise and the quality of her innovative idea often allow her to manipulate the principal by designing how and when the outcomes of the research and development process are announced.

Recently, venture capital firms have started to pour billions into startups focused on the development of quantum computers, which are known for their technological complexity and difficulty of construction. The economic viability
of quantum computing is questioned by a number of experts; however, the startups promise the investors a completed product in the foreseeable future.\footnote{''The Quantum Computing Bubble.'' \emph{Financial Times}, August 25,
2022.} For instance, a quantum startup PsiQuantum announced to potential investors that it hopes to develop a commercially-viable quantum computer within five years and managed to raise more than $\$200$ million in 2019.\footnote{''Bristol Professor's Secretive Quantum Computing Start-Up Raises £179m.'' \emph{The Telegraph}, November 16, 2019.}

This paper studies the implications of the agent's control of information during the progress of a research and development project when the agent and the principal disagree about the timing of when to abandon the research
idea. I ask: What is the degree of transparency to which an agent should commit before starting to work on an innovative project? In particular, which terms for self-reporting on the progress of the project should a startup propose while discussing the term sheet with a venture capitalist? As I show, depending on the properties of the project, the startup would strategically choose both the timing for the disclosure of updates on the progress of the project and the type of news it discloses - either good or bad.

I study a game between a startup and an investor. The startup controls the information on the progress of the project and has the power to propose the terms for self-reporting on it to the venture capitalist.\footnote{I discuss the reasoning behind this assumption in Section \ref{sect:disc_assumpt}.}
The startup has an
intertemporal commitment power and commits to a dynamic information policy, which can be
interpreted as designing the terms of the contract specifying how the information on the progress of the project is disclosed over time as the project unfolds. In return, the investor continuously provides funds and chooses when
to stop funding the project.

The project has two stages and evolves stochastically over time toward completion, conditional on continuous investment in it. The completion of each of the stages of project occurs according to a Poisson process. The completion of the first stage serves as a milestone, such as the development of a prototype, while completion of the second stage achieves the project. The investor gets a lump-sum project completion profit if and only if he stops investing after the project is completed and before an exogenous project completion deadline, and the startup prefers the principal to postpone stopping the funding.\footnote{I discuss the reasons for the presence of the project completion deadline in Section \ref{sect:setup}.}

As the investor receives the reward only after a prolonged period of investment, he initially invests without being able to see if the investment is worthwhile. Hence, it is individually rational for the investor to start investing only if he is sufficiently optimistic regarding the future of the project. An important feature of the setting that I consider is that the \emph{information is symmetric at the outset}: not only the investor, but also the startup is unable to find out if the project will bring profit to the investor or not - this can be inferred only as time goes on and some evidence is accumulated. The only tool that the startup has for persuading the investor to start investing is the promise of future reports on the progress of the project. 

Clearly, the good news about the completion of the project is valuable to the investor as it helps him to stop investing in a timely manner. Further, as evidence regarding the project accumulates over time, failure to pass the milestone in a reasonable time makes the project unlikely to be accomplished in time - and the investor prefers to stop investing after observing such bad news. When designing the information policy, the startup chooses optimally between the provision of these two types of evidence in order to postpone the investor's stopping decision for as long as possible.

I show that the startup's choice of information policy depends on the ex ante attractiveness of the project for the investor. The attractiveness is captured by the \emph{flow cost-benefit ratio of the project}.
Thus, a project is relatively more attractive ex ante to the investor when its flow investment cost is lower, its project completion profit is higher, or the Poisson rate, at which completion of one stage of the project occurs, is higher.

When the project is sufficiently attractive ex ante to the investor, promises to provide information only on the completion of the project serve as a sufficiently strong incentive device to motivate the investor to start the funding at the outset. Further, the future news on the completion of the project does not harm the total expected surplus generated by the interaction of the startup and investor, while the future news on the project being poor decreases the surplus that the startup can potentially extract from the investor. Accordingly, the startup commits to providing only the good news, but not the bad news on the project in the future: it \emph{discloses the completion of the project and postpones the disclosure} in order to ensure the extraction of as much surplus as possible from the investor. In the context of quantum computing, the startup optimally chooses and announces to the venture capitalist the date by which it plans to have a fully developed quantum computer. When the date comes, the startup reports completion if the quantum computer has been completed; if not, the startup reports the completion as soon as it occurs. 

The situation changes when the project does not look promising to the investor ex ante. In that case, if the startup commits to disclosing only the completion of the project, the investor will not have the sufficient motivation to start investing in it. Thus, the startup extends the information policy to encompass not only the good news but also the bad. As in the case of the promising project, the startup discloses the project's completion and does so without any postponement, thereby fully exploiting its preferred incentive tool. In addition, the startup sets a date at which the bad news is released if the milestone of the project has not yet been reached - this date is \emph{the interim reporting deadline}.

Setting the interim deadline, the startup chooses a deterministic date, which it optimally postpones. As the startup prefers the investor to postpone stopping the funding, it prefers the interim deadline to be at a later expected date. Further, the completion of the stages of project according to a Poisson process makes both the startup and the investor risk-averse with respect to the date of the interim deadline. Thus, the startup prefers to set the interim deadline at a deterministic date and to postpone it as late in time as possible in order to extract all the surplus from the investor. In the context of quantum computing, the startup optimally chooses and announces a fixed date by which it plans to have a prototype of the quantum computer. When the date comes, reporting having the prototype at hand convinces the investor to continue the funding, and reporting not having the prototype leads to termination of the project.

Finally, I demonstrate that the outlined structure of the optimal information disclosure holds for a broad class of preferences of the startup and the investor. I allow for profit-sharing between the startup and the investor, varying degrees of the startup's benefit from the flow of funding, and exponential discounting, and show that the startup prefers not to set any interim deadlines whenever the project is sufficiently promising to the investor. The future disclosure of the completion of the project promises investor profit in exchange for a prolonged investment, while the disclosure of the stagnation of the project at the interim deadline promises investor only saved costs, as further investment stops. Thus, when the project is attractive, the startup can make the funding and the beneficial experimentation relatively longer by setting no interim deadlines.


\section{Related literature}

My paper is related to the \emph{literature on dynamic information design}. The closest paper in this strand of literature is by \citet{ely2020moving}. Similarly to my paper, they study the optimal persuasion of a receiver facing a lump-sum payoff to incur costly effort for a longer time. In my model, as in theirs, the sender is concerned to satisfy the beginning-of-the-game individual rationality constraint of the receiver when choosing the information policy. Further, the ``leading on'' information policy in \citet{ely2020moving} has a similar flavor to the ``postponed disclosure of completion'' information policy in my paper: promises of news on completion of the project serve as an incentive device sufficient to satisfy the receiver's individual rationality constraint.

However, there are several substantial differences between \citet{ely2020moving} and my paper. While in their model the state of the world is static and drawn at the beginning of the game, in my model it evolves endogenously over time, given the receiver's investment. As a result, the initial disclosure used in the “moving goalposts” policy in \citet{ely2020moving} cannot provide additional incentives for the receiver in my model. The sender in my model uses another incentive device to incentivize the receiver to opt in at the initial period: she commits to an interim deadline at which she discloses that the first stage of the project is not completed.

Another closely related paper is by \citet{orlov2020persuading}. The
main similarity to my paper lies in the sender's incentive to postpone
the receiver's irreversible stopping decision. The sender in their
paper prefers to backload the information provision, which is in
line with the properties of the optimal information policy in my paper. However, there are a
number of substantial differences between our papers. In \citet{orlov2020persuading},
the sender does not have the intertemporal commitment power. Further,
the receiver obtains a payoff at each moment of time, and thus the sender does
not need to persuade the receiver to opt in at the beginning of the
game.

\citet{ely2017beeps,renault2017optimal,ball2019dynamic} also analyze
dynamic information design models. However, their papers focus on
persuading a receiver who chooses an action and receives a payoff at each moment
of time, whereas in my paper the receiver takes
an irreversible action and receives a lump-sum project completion payoff. \citet{henry2019research} consider a dynamic
Bayesian persuasion model in which, similarly to my model, the receiver
needs to take an irreversible decision. However, the incentives of
the sender and receiver differ from my model: the receiver wants to
match the static state of the world and the sender is concerned with
both the receiver's action choice and with the timing of that choice. \citet{basak2020panics} study dynamic information design in a regime change game. The optimal information policy in their model resembles the interim deadline policy in my model: at a fixed date, the principal sends the recommendation to attack if the regime is substantially weak by that time.

My paper is also related to the \emph{literature on the dynamic provision
of incentives for experimentation} \citep{bergemann1998venture,bergemann2005financing,curello2020screening,madsen2020designing}.
The closest papers in this strand of literature are by \citet{green2016breakthroughs} and \citet{wolf2017informative}.
Similarly to my model, both papers consider design of a contract regarding a two-stage project, in which the completion of stages arrives at a Poisson rate. In \citet{green2016breakthroughs}, there is no project completion deadline and the quality of the project is known to be good, while in \citet{wolf2017informative} the quality of the project is uncertain. In contrast to my paper, both papers focus on a canonical moral-hazard problem and give the power to design the terms of the contract to the investor (principal) rather than the startup (agent). In particular, the contract in \citet{green2016breakthroughs} specifies the terms for the agent's reporting on the completion of the first stage of the project.
Similarly to my model, the optimal
reporting takes the form of a deterministic interim deadline: at a principal-chosen date, the agent truthfully reports if she has already completed the first stage, which determines the further funding of the project.\footnote{In a broad sense, my paper also relates to the small strand of theoretical
literature on dynamic startup-investor and startup-worker relations
under information asymmetry \citep{kaya2020paying,ekmekci2020learning}.
However, while these papers focus on the signaling of
the type of startup, I study the provision of information by the startup on
the progress of the project.}

\section{The model}

\subsection{Setup}\label{sect:setup}
I consider a game between an agent (she, sender) and a principal (he,
receiver). Time is continuous and there is a publicly observable deadline
$T$, $t\in\left[0,T\right]$.\footnote{The results for the setting without a deadline are easily obtained
by considering $T\rightarrow\infty$. They are presented in Appendix
\ref{section_Tinfty}.} 
For each $t$, \emph{the principal} chooses sequentially to invest
in the project $\left(a_{t}=1\right)$ or not $\left(a_{t}=0\right)$.
The flow cost of the investment is constant and given by $c$. The
choice of $a_{t}=0$ at some $t$ is irreversible and induces the
end of the game.\footnote{The absence of the principal's commitment to an investment policy and the irreversibility of the stopping decision capture the venture capitalist's option to abandon the project, e.g., in the case of its negative net present value.}

The assumption that the project needs to be completed in finite time is natural in many economic settings. The main interpretation for $T$ is an expected change in market conditions that renders the project unprofitable. In the context of a research and development project, $T$ could stand for the date at which the competitor's innovative product is expected to enter the market, or the date at which the competitor is expected to get a patent on the competing innovation.

The \emph{state of the world} at time $t$ is captured by the number
of stages of the project completed by $t$, $x_{t}$, and the project
has two stages, $x_{t}\in\{0,1,2\}$. The state process begins at
the state $x_{0}=0$ and, conditional on the continuation of the investment
by the principal, it increases at a Poisson rate $\lambda>0$. 
Information on the number of stages completed is controlled by \emph{the
agent}. Thus, when making investment decisions, the principal relies
on the information provided by the agent.

The project brings the profit $v$ if and only if the \emph{second stage of the project}
has been completed by the time of stopping, and a payoff of $0$,
otherwise. 
I assume that all of the profit goes to the principal. This assumption simplifies the exposition without affecting the main results of the paper. I relax this assumption and consider the profit-sharing between the agent and the principal in Section \ref{sec:prof_share}.

There is a \emph{conflict of interest} between the agent and the principal as the agent benefits from using the funds provided by the principal for running the project, possibly diverting them for her benefit. Thus, the agent faces the flow payoff of
$c$ and wants the principal to postpone his irreversible decision to stop as long as possible.

I study the agent's choice of information provision to the principal.
The agent chooses an information policy to maximize her expected long-run
payoff. I assume that the agent has the power to announce and commit
to a policy. An \emph{information policy} $\sigma$ is a rule that for each
date $t$ and for each past history $h\left(t\right)$ specifies
a probability distribution on the set of messages $M$. The history includes all past and current realizations of the process and all past message draws and principal's action choices. 

When choosing
an information policy, the agent faces a rich strategy space. First,
she can choose if the information on the completion of the first,
or second, stage of the project will be disclosed by the policy. Second,
she can choose how to disclose the completion of a stage of the project:
for instance, to do so immediately or to postpone the disclosure.

The timing of the game is as follows. First, at $t=0$, the agent
publicly commits to an information policy $\sigma$. Second, at each
$t$ the principal observes the message generated by the information
policy and makes her investment decision. The game ends when the principal
chooses to stop investing or at $T$, if he keeps investing until
$T$. I assume that whenever indifferent about investing or not, the
principal chooses to invest, and whenever indifferent about disclosing
information or not, the agent chooses not to disclose.

Throughout the paper, I use the following intuitive notational convention: for any two dates at
which the principal stops investing, $S$ and $\tau$,

\begin{equation*}
\begin{aligned}
 S\wedge\tau & \coloneqq\min\left\{S,\tau\right\},\\
S\lor\tau & \coloneqq\max\left\{S,\tau\right\}. 
\end{aligned}    
\end{equation*}

\subsection{Discussion of assumptions}\label{sect:disc_assumpt}
The main interpretation of the considered dynamic information design
problem is the contracting between the agent (startup) and the principal (investor) on
the terms of reporting on the completion of stages of the project
that are not publicly observed. The terms could take the form of a
proposed formal reporting schedule or a schedule 
of meetings with the investor. Non-observability of the stage completions
stems from the fact that, while the technology is being developed in the
lab, the principal either does not have sufficient expertise
to assess the progress or the full access to the
lab.

I assume that the principal does not have the power to propose the
terms for reporting to the agent and, e.g., make her fully disclose
the progress achieved in the lab. The most natural interpretation of such
an asymmetry in the bargaining power is the asymmetry in the market for private equity: there are sufficiently many
investors willing to invest in a particular technology or sufficiently few startups working on the technology.\footnote{In the alternative interpretation of the model, contracting concerns internal corporate
research and development and takes place between the leading researcher
and the headquarters of a company. The leading researcher's bargaining
power in proposing the terms for disclosure again stems from the market
asymmetry: the specialists having the desired level of expertise might
be in a short supply.} For instance, investors' interest in quantum computing has grown markedly in recent years, while there are reports of a shortage of human capital in this industry.\footnote{''The Quantum Computing Bubble.'' \emph{Financial Times}, August 25,
2022.}\footnote{``Quantum Computing Funding Remains Strong, but Talent Gap Raises Concern'', a report by McKinsey Digital, https://www.mckinsey.com/business-functions/mckinsey-digital/our-insights/quantum-computing-funding-remains-strong-but-talent-gap-raises-concern/.} Another example is the communication software industry, which has recently experienced increased investment activity.\footnote{''This Is Insanity: Start-Ups End Year in a Deal Frenzy.'' \emph{Best Daily Times}, December 07,
2020.}

As the agent enjoys the power of full control over the information
on the progress of the project, she is completely free to offer what
is disclosed and when. In particular, the contract between the agent
and the principal can specify that the completion of the second stage
of the project is disclosed with a delay rather than immediately.
The agent who has an advantage in expertise over the principal can
rationalize such a condition by saying that before the success is
reported to the principal, it is worth re-checking the data, which
takes time.

Even though the principal can not dictate to the agent which information
and how she should disclose, the principal can potentially hire an
external monitor who would visit the lab and prepare an additional
report on the progress of the project. In that case, the contract
signed between the agent and the principal will account for both free
information that the agent promised to provide and additional costly
information which the principal obtains with the help of a monitor.
In the baseline version of the model, I assume that the principal
can not use the help of a monitor. This can be rationalized by the
shortage of experts in the field, which makes hiring a monitor prohibitively
costly. Alternative interpretation is that the agent restricts the principal's access
to additional information on the progress of the project by stating
that a potential information leak would put the technology being developed
at risk.\footnote{In particular, this rationale was used to restrict the investors'
access to information on the progress of the project in the case of
Theranos, see ''What Red Flags? Elizabeth Holmes Trial Exposes Investors’ Carelessness.'' \emph{The New York Times}, November 04,
2021.}

The information policy relies upon the agent's commitment power, which
holds not only within each date but also between the dates. The agent's
commitment within each date follows from prohibitively high legal
costs of cooking up the evidence. The agent's intertemporal commitment
stems from the rigidity of terms and form of reporting fixed in the
contract that the agent and the principal sign at the outset of the game.

\section{No-information and full-information benchmarks}

\subsection{No-information benchmark}

\label{section_no_inf}

First, I consider the simple case when the information policy is given
by $\sigma^{NI}$: the same message $m$ is sent for all histories
$h\left(t\right)$ and all dates $t$. Thus, the agent provides
no information regarding the progress of the project. As I demonstrate,
in this case the principal starts investing in the project if and
only if it is sufficiently promising for the principal from the ex
ante perspective and invests until a deterministic interior date.

As no news arrives, the principal bases his decision about when to
stop investing on his unconditional belief regarding the completion
of the second stage of the project. I denote the unconditional belief
that $n$ stages of the project were completed by $t$, by $p_{n}\left(t\right)$,
i.e., $p_{n}\left(t\right)\coloneqq\p\left(x_{t}=n\right)$. The state
of the world is fully determined by $p\left(t\right)$ given by 
\[
\begin{aligned}p_{0}\left(t\right) & =e^{-\lambda t},\\
p_{1}\left(t\right) & =\lambda te^{-\lambda t},\\
p_{2}\left(t\right) & =1-e^{-\lambda t}-\lambda te^{-\lambda t}.
\end{aligned}
\]

The principal's sequential choice of $a_{t}\in\left\{ 0,1\right\} $
can be restated equivalently as the choice of deterministic stopping
time $S^{NI}\in\left[0,T\right]$ chosen at $t=0$.\footnote{Note that the dynamic belief system that he faces is deterministic
in a sense of being fully specified from $t=0$ perspective.} Given the principal's continuous investment, the probability of completion
of the second stage of the project, $p_{2}\left(t\right)$, increases
monotonously over time, making obtaining the payoff $v$ more likely.
However, postponing the stopping is costly.

To decide on $S^{NI}$, the principal trades off the flow benefits
and flow costs of postponing the stopping decision, while keeping
the individual rationality constraint in mind. The flow cost of postponing
the stopping for $\Delta t$ is given by $c\cdot\Delta t$ and the
flow benefit is given by $v\cdot p_{1}\left(t\right)\lambda\Delta t$.\footnote{To observe this, note that the probability of the completing both
the first and second stages within a very short time $\Delta t$ is
negligibly small; thus, during some $\Delta t$, the principal receives
the project completion payoff $v$ iff the first stage has already
been completed.} Thus, the necessary condition for the principal's stopping at some
interior moment of time $\left(0<S<T\right)$ is given by

\begin{equation}
v\cdot p_{1}\left(S\right)\lambda=c.\label{FOC_no_info}
\end{equation}
Let 
\[
\kappa\coloneqq\frac{c}{v\lambda},
\]
the ratio of the flow cost of investment $c$ to the gross project
payoff $v$ normalized using $\lambda$, the rate at which a project
stage is completed in expectation. The intuitive interpretation of
$\kappa$ is the \emph{flow cost-benefit ratio of the project}.
$\kappa$ is an inverse measure of how ex ante promising the project
is for the principal. (\ref{FOC_no_info}) is equivalently given by\footnote{Here I WLOG express the flow benefits and flow costs of investing
for the principal in different units of measurement.}

\begin{equation}
\underbrace{p_{1}\left(S\right)}_{\text{flow benefit of waiting}}=\underbrace{\kappa}_{\text{flow cost of waiting}}\label{no_feedb_necess_cond}
\end{equation}
and presented graphically in Figure \ref{noinfo_figure}. As the state
process transitions monotonously from $0$ to $1$ and then to $2$,
$p_{1}\left(t\right)$ first increases and after some time starts
to decrease. Thus, the principal has two candidate interior stopping
times satisfying (\ref{no_feedb_necess_cond}), $\bar{S}$ and $\bar{S}^{NI}$.
The principal prefers to postpone stopping to $\bar{S}^{NI}$, as during
$\left(\bar{S},\bar{S}^{NI}\right)$ the flow benefits are higher than the
flow costs.

\begin{figure}[htb]
\captionsetup{justification=centering} \centering %
\noindent\begin{minipage}[c]{0.51\textwidth}%
\centering \begin{tikzpicture}[xscale = 1.9,yscale = 5.2,domain=0:3.2]
\draw[thin,->] (0,0) node[below]{$0$} -- (3.5,0) node[below]{$t$}; 
\draw[thin,->] (0,0) -- (0,0.98); 
\draw[thin] (-0.02,0.9) -- (0.02,0.9) node[left]{$1$};
\draw[thick,red!50!black] plot (\x, {1-exp(-\x)-\x*exp(-\x)}); 
\draw[thick,red!50!black,-stealth] (1.5,0.66) node[above]{
\begin{tabular}{cc}
$p_{2}(t)$, probability of completion \\
of 2nd stage of project \\
\end{tabular}
} -> (1.5,0.5);
\end{tikzpicture} %
\end{minipage}%
\noindent\begin{minipage}[c]{0.51\textwidth}%
\centering \begin{tikzpicture}[xscale = 1.9,yscale = 10,domain=0:3.2,samples = 200]
\draw[thick,->] (0,0) node[below]{$0$} -- (3.5,0) node[below]{$t$}; 
\draw[thick,->] (0,0) -- (0,0.48); 
\draw[thick] (-0.03,1/e) node[left]{$\frac 1 e$} -- (0.03,1/e);
\draw[thick,red!50!black,-stealth] (1.45,0.37) node[above]{
\begin{tabular}{cc}
$p_{1}(t)$, flow benefit \\
of waiting \\
\end{tabular}
} -> (1.6,0.34);
\draw[thick,black,-stealth] (1.4,0.18) node[below]{
\begin{tabular}{cc}
$\frac{c}{v \lambda}$, flow cost \\
of waiting \\
\end{tabular}
} -> (1.5,0.21);
\draw[thick,black,-stealth] (2.8,0.25) node[above]{
\begin{tabular}{cc}
optimal \\
choice \\
\end{tabular}
} -> (2.375,0.228);
\draw[thick,red!50!black] plot (\x,{\x * exp(-\x)}); 
\draw[thin] (0,0.22) -- (3.5,0.22);
\draw[ultra thin,dashed] (0.295692,0) node[below]{$\bar{S}$} -- +(0,0.45);
\draw[ultra thin,,dashed] (2.38212,0) node[below]{$\bar{S}^{NI}$} -- +(0,0.45);
\end{tikzpicture} %
\end{minipage}\caption{Principal's choice under no information:\protect \protect \\
 \textbf{left plot:} postponing stopping increases the chance of getting
a project payoff $v$;\protect \protect \\
 \textbf{right plot:} principal trades off costs and benefits and
optimally stops at $\bar{S}^{NI}$.}
\label{noinfo_figure} 
\end{figure}
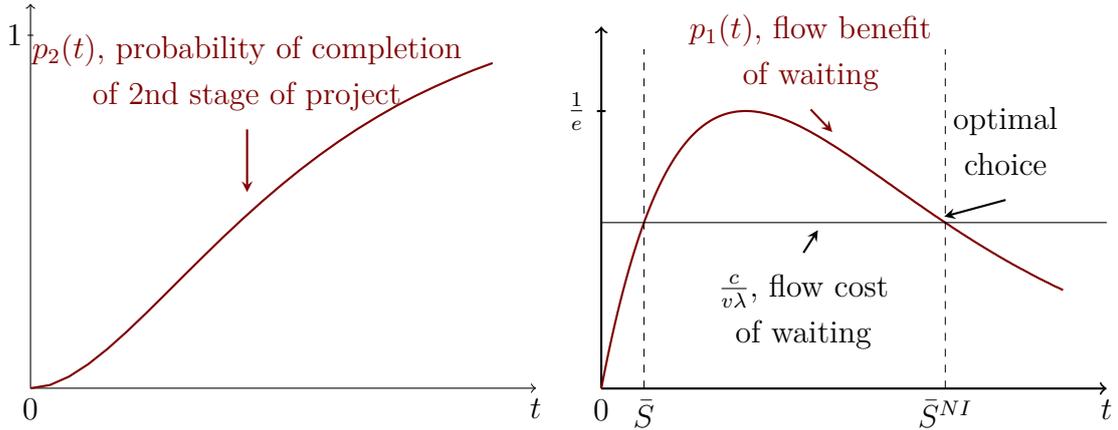

The forward-looking principal can guarantee himself a payoff of $0$
if he does not start investing at $t=0$. Thus, he will choose to
start investing at $t=0$ only if his flow gains accumulated up to
$T\wedge \bar{S}^{NI}$ are larger than his flow losses, and his expected payoff is given by 
\begin{equation}
V^{NI}\coloneqq\max\left\{0,\int_{0}^{T\wedge \bar{S}^{NI}}\left(v\cdot p_{1}\left(s\right)\lambda-c\right)ds\right\}.\label{NI_integral}
\end{equation}
Geometrically, the integral in (\ref{NI_integral}) represents the difference between the shaded areas in Figure \ref{noinfo_figure_2}
that correspond to the accumulated gains and losses. The principal
starts investing at $t=0$ if, given $T$ and $\lambda$, the normalized
cost-benefit ratio $\kappa$ is low enough, so that the shaded area
of the accumulated gains is at least as large as that of the accumulated
losses. I denote such a cutoff value of $\kappa$ by $\kappa^{NI}\left(T,\lambda\right)$
and summarize the principal's choice without information in Lemma
\ref{no_feedb}.

\begin{figure}[htb]
\captionsetup{justification=centering} \centering %
\noindent\begin{minipage}[c]{0.51\textwidth}%
\centering \begin{tikzpicture}
\begin{axis}[
    axis lines = middle,
    xtick = {2.19},
    xticklabels = {$\bar{S}^{NI}$},
    ytick = {0},
    xlabel = {$t$},
    ylabel = {},
    x label style={at={(axis description cs:1,-0.1)}},
    xmin=0, xmax=3.2,
    ymin=0, ymax=0.5]
 
\addplot [name path = A,
    domain = 0:4.5,
    samples = 1000] {x * exp(-x)} 
    node [very near end, right] {$y=x^2$};
 
\addplot [name path = B,
    domain = 0:4.5] {0.25} 
    node [pos=1, below] {$y=x$};
 
\addplot [teal!30] fill between [of = A and B, soft clip={domain=0:2.19}];

\draw[ultra thin,dashed] (2.8,0.25) node[above]{$\kappa \coloneqq \frac{c}{v \lambda}$} -- (2.8,0.25); 

\draw[ultra thin,dashed] (3,0.16) node[above]{$p_{1}(t)$} -- (3,0.16); 
 
\draw[ultra thin,dashed] (2.19,0)
-- (2.19,0.4); 

\draw[thick,black,-stealth] (1,0.39) node[above]{
\begin{tabular}{cc}
accumulated \\
gains \\
\end{tabular}
} -> (1,0.33);

\draw[thick,black,-stealth] (0.7,0.15) node[below]{
\begin{tabular}{cc}
accumulated \\
losses \\
\end{tabular}
} -> (0.05,0.17);
 
\end{axis}
\end{tikzpicture} %
\end{minipage}%
\noindent\begin{minipage}[c]{0.51\textwidth}%
\centering \begin{tikzpicture}
\begin{axis}[
    axis lines = middle,
    xtick = {1.6,2.19},
    xticklabels = {$T$,$\bar{S}^{NI}$},
    ytick = {0},
    xlabel = {$t$},
    ylabel = {},
    x label style={at={(axis description cs:1,-0.1)}},
    xmin=0, xmax=3.2,
    ymin=0, ymax=0.5]
 
\addplot [name path = A,
    domain = 0:4.5,
    samples = 1000] {x * exp(-x)} 
    node [very near end, left] {$y=x^2$};
 
\addplot [name path = B,
    domain = 0:4.5] {0.25} 
    node [pos=1, below] {$y=x$};
 
\addplot [teal!30] fill between [of = A and B, soft clip={domain=0:1.6}];
 
\draw[ultra thin,dashed] (2.8,0.25) node[above]{$\kappa \coloneqq \frac{c}{v \lambda}$} -- (2.8,0.25); 

\draw[ultra thin,dashed] (3,0.16) node[above]{$p_{1}(t)$} -- (3,0.16); 
 
\draw[ultra thin,dashed] (2.19,0) 
-- (2.19,0.4); 

\draw[ultra thin,dashed,red] (1.6,0) 
-- (1.6,0.4); 

\draw[thick,black,-stealth] (1,0.39) node[above]{
\begin{tabular}{cc}
accumulated \\
gains 
\end{tabular}
} -> (1,0.33);

\draw[thick,black,-stealth] (0.7,0.15) node[below]{
\begin{tabular}{cc}
accumulated \\
losses 
\end{tabular}
} -> (0.05,0.17);
 
\end{axis}
\end{tikzpicture} %
\end{minipage}\caption{Principal's choice to start investing at $t=0$ or not:\protect \protect \\
 \textbf{left plot:} $T>\bar{S}^{NI}$; the project deadline is distant and
decision-irrelevant;\protect \protect \\
 \textbf{right plot:} $T\protect\leq \bar{S}^{NI}$; the project deadline
is close, which leads to lower expected benefits of investing.}
In both plots the expected accumulated gains are higher than the
losses, so the principal starts to invest at $t=0$. \label{noinfo_figure_2} 
\end{figure}
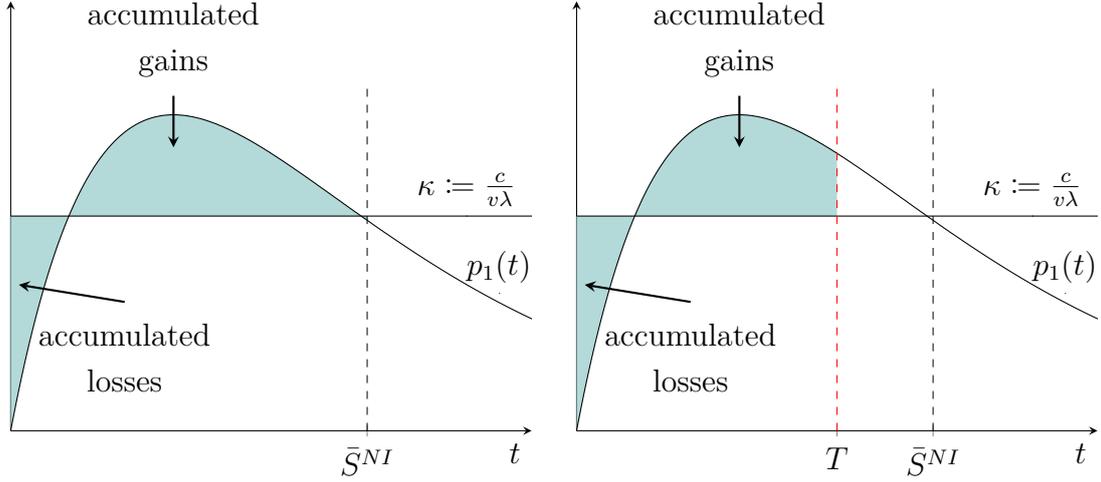

\begin{lemma}\label{no_feedb} Assume no information regarding the
progress of the project arrives over time. Denote the time at which
the principal stops investing by $S^{NI}$. If $\kappa>\kappa^{NI}\left(T,\lambda\right)$,
then the principal does not start investing in the project, i.e., $S^{NI}=0$.
If $\kappa\leq\kappa^{NI}\left(T,\lambda\right)$, then the principal's
choice of stopping time is given by

\begin{equation}
S^{NI}=\begin{cases}
\bar{S}^{NI}, & \text{if }\frac{1}{\lambda}\leq T\text{ and }\kappa\geq e^{-\lambda T}\lambda T\\
T, & \text{otherwise ,}
\end{cases}\label{T_NI}
\end{equation}
the closed-form expressions for $\bar{S}^{NI}$ and $\kappa^{NI}\left(T,\lambda\right)$
are presented in the proof in Appendix \ref{app-proofs}.

\end{lemma}

\subsection{Full-information benchmark}

\label{section_FI}

Here, I consider the case in which the information policy is given
by $\sigma^{FI}$: $M=\left\{ m_{0},m_{1},m_{2}\right\}$
and the message $m_{n}$ is sent for all $t$ such that $x_{t}=n$,
$n\in\left\{ 0,1,2\right\} $. Thus, the principal fully observes
the progress of the project at each $t$.\footnote{This benchmark corresponds to equilibrium in the setting, where the principal has the full power to propose the terms of self-reporting to the agent.} I characterize the cutoff
level of the cost-benefit ratio below which the principal is willing
to start investing. Further, I show that the principal chooses to
stop when no stages of the project are completed and the project completion
deadline $T$ is sufficiently close.

At each $t$, the principal uses the information on the number of
stages completed to decide either to stop investing or to postpone
the stopping. The news on completion of the second stage of the project
makes the principal stop immediately, as this way he immediately receives
the project payoff $v$ and stops incurring the costs of further investment.
If only the first stage of the project is completed, the principal
faces the following trade-off. The instantaneous probability that
the second stage will be completed during $\Delta t$ is given by
$\lambda\Delta t$, which is constant over the time. Thus, the expected
benefit of postponing the stopping for $\Delta t$ is given by $v\cdot\lambda\Delta t$.
Meanwhile, the expected cost of postponing the stopping is given by
$c\cdot\Delta t$. As a result, if $\kappa\leq1$, then the principal
who knows that the first stage of the project has already been completed
invests until either the completion of the second stage or until the
project deadline $T$ is reached.

Consider now the case in which the principal knows that the first
stage has not yet been completed. The principal's trade-off with respect
to the stopping decision is now more involved. Postponing the stopping
for $\Delta t$ leads to the completion of the first stage of the
project with the instantaneous probability $\lambda\Delta t$. Completion
of the first stage of the project at some $t$ implies that the principal
receives the continuation value of the game, conditional on having
the first stage completed. I denote the continuation value of the
principal at time $t$ under full information and conditional on the
completion of first stage of the project by $V_{t|1}^{FI}$.
This is given by\footnote{See the derivation in the proof of Lemma \ref{lemma_FI} in the Appendix.}
\begin{equation}
V_{t|1}^{FI}=\left(v-\frac{c}{\lambda}\right)\left(1-e^{-\lambda\left(T-t\right)}\right).\label{V_1_FI}
\end{equation}
The principal's expected benefit from postponing the stopping for
$\Delta t$ is given by $V_{t|1}^{FI}\cdot\lambda\Delta t$
and the cost of postponing the stopping is, as before, given by $c\cdot\Delta t$.
The continuation value, $V_{t|1}^{FI}$, shrinks over
time and approaches $0$ as the project deadline $T$ approaches.
This is because the shorter the time left before the project deadline,
the less likely it is that the second stage of the project will be
completed before $T$. If at some $t$, and given that no stages are
completed yet, the expected net benefit of investing reaches $0$,
it is optimal for the principal to stop at that $t$.\footnote{If at $t$ the expected benefit of investing becomes lower than the
cost, then, after $t$, the net expected benefit remains negative.
Thus, it is optimal for the principal to stop investing precisely
at $t$.} I denote this date by $S_{0}^{P}$ and plot it in Figure \ref{figure_FI}.

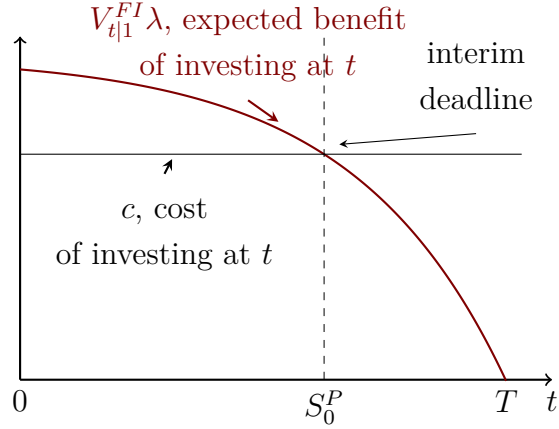
\begin{figure}[H]
\captionsetup{justification=centering} \centering \begin{tikzpicture}[xscale = 2,yscale = 1.3,domain=0:3.2,samples = 200]
\draw[thick,->] (0,-2.3) node[below]{$0$} -- (3.5,-2.3) node[below]{$t$}; 
\draw[thick,->] (0,-2.3) -- (0,1.2); 
\draw[thick,red!50!black] plot (\x,{1- exp(\x - 2)}) node[left,above]{};
\draw[thin] (0,0) -- (3.3,0);
\draw[ultra thin,dashed] (2,-2.3) node[below]{$S_{0}^{P}$} -- (2,1.2);
\draw[ultra thin] (3.2,-2.3) node[below]{$T$} -- (3.2,-2.3);
\draw[thick,red!50!black,-stealth] (1.5,0.55) node[above]{
\begin{tabular}{cc}
$V_{t|1}^{FI} \lambda$, expected benefit \\
of investing at $t$ \\
\end{tabular}
} -> (1.7,0.33);
\draw[thick,black,-stealth] (0.95,-0.2) node[below]{
\begin{tabular}{cc}
$c$, cost \\
of investing at $t$ \\
\end{tabular}
} -> (1,-0.05);
\draw[black,-stealth] (3,0.21) node[above]{
\begin{tabular}{cc}
interim \\
deadline
\end{tabular}
} -> (2.1,0.1);
\end{tikzpicture} \caption{The principal optimally sets an interim deadline $t=S_{0}^{P}$ under
full information: given that the first stage of the project has not
been completed by $S_{0}^{P}$, it is optimal to stop investing at
$S_{0}^{P}$.}
\label{figure_FI} 
\end{figure}

As the principal has an incentive to stop at $S_{0}^{P}$ only if
he knows that the first stage or the milestone of the project has
not been reached, the economic interpretation of $S_{0}^{P}$ is
that it is \emph{the interim deadline} that the principal sets for
the project. If the milestone has not been reached by the interim
deadline, then it is sufficiently unlikely that the project will be
completed before the project deadline $T$. Thus, it is optimal for
the principal to ``give up'' on the project and stop investing at
$t=S_{0}^{P}$. If the milestone is reached by the interim deadline,
then the principal has an incentive not to stop investing until either
the second stage is completed or $T$ is hit.

Finally, given the plan to stop either at the interim deadline, or
at the completion of the second stage of the project, it is individually
rational to start investing only if the principal's expected payoff
from the $t=0$ perspective is non-negative. I denote the upper bound
for the normalized cost-benefit ratio such that the principal starts
investing at $t=0$ by $\kappa^{FI}\left(T,\lambda\right)$. Intuitively, $\kappa^{FI}\left(T,\lambda\right)>\kappa^{NI}\left(T,\lambda\right)$: whenever the principal is willing to start investing under no information, he is also willing to start under the full information. I summarize
the principal's choice under full information in Lemma \ref{lemma_FI}.
\begin{lemma}\label{lemma_FI} Assume that the progress of the project
is fully observable at each moment in time. If $\kappa>\kappa^{FI}\left(T,\lambda\right)$, where $\kappa^{FI}\left(T,\lambda\right)>\kappa^{NI}\left(T,\lambda\right)$,
then the principal does not start investing in the project. If $\kappa\leq\kappa^{FI}\left(T,\lambda\right)$,
the principal invests either until the random date at which the second
stage of the project is completed, $t=\tau_{2}$, or until the interim
deadline, $t=S_{0}^{P}$, at which he stops if the first stage has
not yet been completed. Formally, the time at which the principal
stops investing is a random variable $\tau$ given by:

\[
\tau=\begin{cases}
\tau_{2}\wedge T, & \text{if }x_{S_{0}^{P}}\neq0\\
S_{0}^{P}, & \text{otherwise },
\end{cases}
\]
where $S_{0}^{P}=T+\frac{1}{\lambda}\log\left(\frac{1-2\kappa}{1-\kappa}\right)$
and the expression for $\kappa^{FI}\left(T,\lambda\right)$ is presented
in the proof in Appendix \ref{app-proofs}.

\end{lemma}

Assume now that the agent chooses which information to provide to
the principal. As for $\kappa>\kappa^{FI}\left(T,\lambda\right)$
the principal is not willing to start investing even under full information,
there is no way in which the agent can strategically conceal the information
to her benefit. In Section \ref{stoppingsection}, I assume $\kappa\leq\kappa^{FI}\left(T,\lambda\right)$
and analyze how the agent can strategically provide information on
the progress of the project and extract the principal's surplus.

\section{Agent's choice of information policy}
\label{stoppingsection}

In this Section, I present how the agent's choice of information policy
changes with the ex ante attractiveness of the project, which is captured
by the cost-benefit ratio $\kappa$. 
In Section \ref{sect:opt_discl}, I start with Proposition \ref{prop:opt_discl} which summarizes the comparative statics result. In Sections \ref{case1section}-\ref{section_INT}, I proceed with the detailed discussion of the
economic mechanisms that determine the outlined structure of the optimal information policy. 
Throughout Section \ref{stoppingsection}, I maintain the following technical assumption:
\begin{assumption}\label{assumpt1}
$e^{\lambda T}>\lambda T\left(\lambda T+1\right)+1$.
\end{assumption}
For the sake of a clearer exposition, this assumption rules out the case in which $T$ is so low that whenever the principal is willing to start investing in the no-information benchmark, he invests until $T$.  Relaxing this assumption does not change
the the comparative statics result in Proposition \ref{prop:opt_discl} qualitatively.\footnote{I discuss the implications of
relaxing this assumption in the proof of Proposition \ref{prop:opt_discl}.}

\subsection{The structure of optimal information disclosure}\label{sect:opt_discl}

\begin{proposition}\label{prop:opt_discl}
There exist cost-benefit ratio cutoffs $\kappa^{ND}\left(T,\lambda\right),\kappa^{ND}\left(T,\lambda\right)<\kappa^{NI}\left(T,\lambda\right)$,
and $\tilde{\kappa}\left(T,\lambda\right),\kappa^{NI}\left(T,\lambda\right)<\tilde{\kappa}\left(T,\lambda\right)<\kappa^{FI}\left(T,\lambda\right)$,
such that, depending on the cost-benefit ratio of the project, the
optimal information policy has the following form:
\begin{itemize}
\item[1.] when $\kappa\leq\kappa^{ND}\left(T,\lambda\right)$,
the agent provides no information and the principal invests until
$T$;

\item[2.] when $\kappa^{ND}\left(T,\lambda\right)<\kappa\leq\tilde{\kappa}\left(T,\lambda\right)$,
the agent discloses only the completion of the second stage
of the project and does that with the postponement;

\item[3.] when $\tilde{\kappa}\left(T,\lambda\right)<\kappa<\kappa^{FI}\left(T,\lambda\right)$,
the agent immediately discloses the completion of the second stage
of the project whenever it occurs and specifies a deterministic
interim deadline, at which it discloses if the first stage is already
completed;

\item[4.] when $\kappa\geq\kappa^{FI}\left(T,\lambda\right)$, the agent
provides no information as the principal's long-run payoff is non-positive
even under full information.
\end{itemize}
\end{proposition}

Figure \ref{fig_mechanisms} illustrates Proposition \ref{prop:opt_discl}
and presents the partition of the cost-benefit ratio space based on the
corresponding forms of the optimal information policy.

\begin{figure}[H]
\centering 
\hspace*{-0.05\linewidth}
\begin{tikzpicture}[xscale = 4,yscale = 2,domain=0.7:3.3,samples = 200]
\draw[thick,->] (0.3,1) 
-- (3.85,1) node[below]{$\kappa\left(T,\lambda\right)$}; 
\draw[ultra thin,,dashed] (0.31,1) node[below]{$0$} -- (0.31,0.9);
\draw[ultra thin,,dashed] (1.95,1) -- (1.95,2);
\draw[ultra thin,,dashed] (3.25,1) -- (3.25,2);
\draw[ultra thin,,dashed] (0.93,1) -- (0.93,2);
\draw[ultra thin,,dashed] (1.95,1) node[below]{$\tilde{\kappa}$} -- (1.95,0.9);
\draw[ultra thin,,dashed] (3.25,1) node[below]{$\kappa^{FI}$} -- (3.25,0.9);
\draw[ultra thin,,dashed] (1.4,1) node[below]{$\kappa^{NI}$} -- (1.4,0.9);
\draw[ultra thin,,dashed] (0.93,1) node[below]{$\kappa^{ND}$} -- (0.93,0.9);
\draw[ultra thin,,dashed] (3.6,1) node[below]{$\frac{1}{2}$} -- (3.6,0.9);

\draw[thick,red!50!black] (1.43,1.2) node[above]{
\begin{tabular}{cc}
Postponed disclosure \\
of stage 2 completion \\
\end{tabular}
} -- (1.43,1.2);

\draw[thick,red!50!black] (2.6,1) node[above]{
\begin{tabular}{cc}
Immediate disclosure of \\
stage 2 completion and \\
interim deadline for stage 1
\end{tabular}
} -- (2.7,1);

\draw[thick,black] (0.57,1) node[above]{
\begin{tabular}{cc}
Non-disclosure\\(principal\\invests until $T$)
\end{tabular}
} -- (0.45,1);

\draw[thick,black] (3.65,1) node[above]{
\begin{tabular}{cc}
Non-disclosure\\(principal rejects \\ the project)
\end{tabular}
} -- (3.65,1);

\end{tikzpicture} \caption{Comparative statics of the form of optimal information policy with respect to the cost-benefit ratio of the project, $\kappa\left(T,\lambda\right)$.}
\label{fig_mechanisms} 
\end{figure}
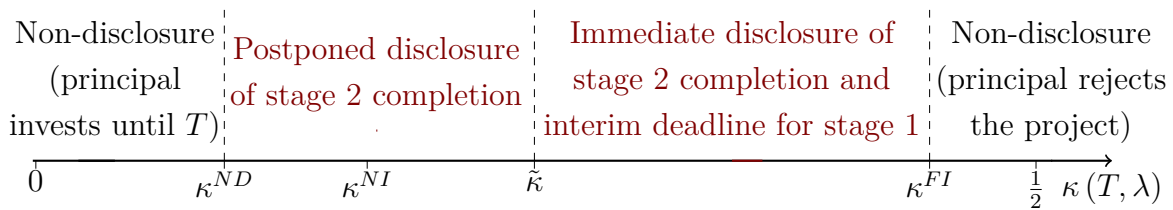

The structure of optimal disclosure presented in Proposition \ref{prop:opt_discl}
follows the simple and intuitive pattern. The lower is the value
of cost-benefit ratio, the higher is ex ante attractiveness of the
project to the principal. First, for $\kappa\leq\kappa^{ND}\left(T,\lambda\right)$,
the project is so attractive that the principal is willing to keep
investing until the project deadline $T$ even in the no-information
benchmark. Thus, there is no need to disclose any information. For
the higher values of $\kappa$, there emerges a room for strategic
disclosure, and the higher is the value of $\kappa$ (i.e., the lower is the ex ante attractiveness of the project),
the more information the agent has to disclose to incentivize the
principal. For $\kappa\geq\kappa^{FI}\left(T,\lambda\right)$, the
project gets so unattractive that the principal can not strictly benefit from investing even in the full-information benchmark. In this extreme case,
the agent chooses not to disclose any information.

The most important part of the result in Proposition \ref{prop:opt_discl}
demonstrates which additional pieces of information the agent chooses
to disclose and when she chooses to discloses them as $\kappa$ gets higher and higher. When $\kappa\in(\kappa^{ND}\left(T,\lambda\right),\tilde{\kappa}\left(T,\lambda\right)]$,
the agent discloses only the completion of the second stage of the project and
does not promise any information on the completion of the first stage
of the project. Further, as $\kappa$ increases from $\kappa^{ND}\left(T,\lambda\right)$ to  $\tilde{\kappa}\left(T,\lambda\right)$, the agent adjusts
the timing of the disclosure: she postpones the disclosure of the second stage completion less and
less and discloses immediately for  $\tilde{\kappa}\left(T,\lambda\right)$. For
$\kappa\in(\tilde{\kappa}\left(T,\lambda\right),\kappa^{FI}\left(T,\lambda\right))$,
the agent not only discloses the completion of the second stage of the project immediately, but also
provides information on the completion of the first stage
at the interim deadline that she optimally chooses.

In the subsequent Sections, I provide details on the mechanisms that shape the aforementioned comparative statics results. I omit the trivial case of non-disclosure under  $\kappa\leq\kappa^{ND}\left(T,\lambda\right)$ and start the discussion from the optimal information policy under $\kappa\in(\kappa^{ND}\left(T,\lambda\right),\tilde{\kappa}\left(T,\lambda\right)]$.

\subsection{Postponed disclosure of project completion}
\label{case1section}

In this Section, I restrict attention to
$\kappa\in(\kappa^{ND}\left(T,\lambda\right),\tilde{\kappa}\left(T,\lambda\right)]$ and explain why the optimal information policy has the particular form presented in the Proposition \ref{prop:opt_discl}: the agent discloses only the completion of the project and does this with the postponement. 

\subsubsection{Agent's problem}
To characterize the agent's choice of information policy, I consider an equivalent problem, in which the agent directly chooses the stochastic history-contingent length of investment subject to the principal's individual rationality constraints that ensure optimality of such action process for the principal.  An \emph{investment
schedule} is a random variable $\tau:  \Omega\rightarrow\ [0,T]$ defined on the probability
space $\left(\Omega,\mathcal{F},\mathbb{P}\right)$ and adapted to the filtration $F=(\mathcal{F}_{t})_{t\geq0}$ generated by the stochastic process $x_{t}$. As I demonstrate in Section \ref{sect:relaxed}, restricting attention to random variables adapted to the \emph{natural} filtration of $x_t$  is without loss of generality for the agent's equilibrium expected payoff when $\kappa\in(\kappa^{ND}\left(T,\lambda\right),\tilde{\kappa}\left(T,\lambda\right)]$.\footnote{In other words, there is no need for external randomization devices to optimally incentivize the principal when $\kappa\in(\kappa^{ND}\left(T,\lambda\right),\tilde{\kappa}\left(T,\lambda\right)]$.}

Informally, an investment schedule $\tau$ is a random variable with support $\left[0,T\right]$
specified by a rule that suggests when to stop investing depending on the history
of previous realizations of the number of completed stages $x_{t}$.\footnote{The stopping rules from the no-information and full-information benchmarks are given in Lemmas \ref{no_feedb} and \ref{lemma_FI}, respectively. Further examples of such rules include ``stop 1 minute after the second stage
of the project is completed'' and ``stop at $t=S$ if only the first
stage of the project is completed by $t=S$''.} The agent chooses this rule at $t=0$.  $\p\left(x_{\tau}=2\right)$ captures the belief about two
stages of the project completed by $\tau$, the random time of stopping in
the future, and $\Esp\left[\tau\right]$ captures the expected length
of investment from $t=0$ perspective.

Given an investment schedule $\tau$, the long-run payoff of the agent
and the principal are given, respectively, by

\[
\begin{aligned}W\left(\tau\right) & \coloneqq\Esp\left[\tau\right]c,\\
V\left(\tau\right) & \coloneqq\p\left(x_{\tau}=2\right)v-\Esp\left[\tau\right]c.
\end{aligned}
\]

As an investment schedule $\tau$ is an action recommendation rule, the action recommendations generated by this rule have to be obedient for the principal. In other words, at each date and for each possible history the principal's actions suggested by $\tau$ have to be optimal for the principal. An object useful for characterizing if an investment schedule $\tau$ generates obedient action recommendations is given by the principal's continuation value at some interim date $t$ promised by the investment schedule $\tau$. This continuation value depends on the beliefs of the principal. 

As the principal does not commit to a policy at $t=0$, he rationally updates his belief given an investment schedule $\tau$ and assesses the costs and benefits of either further following the investment schedule $\tau$ provided by the agent or deviating from it. The absence of stopping by some date $t$ and, thus, the fact that the game continues at $t$ serves as a source of inference for the principal. First, he forms a belief regarding the number of completed stages of the project by $t$, conditional on the game still continuing at $t$, $\p\left(x_{t}=n|t<\tau\right)$.
Second, he forms a belief regarding the number of completed stages
of the project at the random date of stopping in the future, $\tau$, $\p\left(x_{\tau}=n|t<\tau\right)$.

Given the absence of stopping by $t$, the principal's expected payoff
promised by the schedule is given by $\p\left(x_{\tau}=2|t<\tau\right)v-\Esp\left[\tau-t|t<\tau\right]c$.
The principal's expected payoff from stopping at $t$ is given by
$\p\left(x_{t}=2|t<\tau\right)v$. The principal's \emph{continuation
value at $t$} given the investment schedule $\tau$ is the difference
between these two expected payoffs, I denote it by $V_{t}\left(\tau\right)$:

\begin{equation}
V_{t}\left(\tau\right)\coloneqq \left[\p\left(x_{\tau}=2|t<\tau\right)-\p\left(x_{t}=2|t<\tau\right)\right]v-\Esp\left[\tau-t|t<\tau\right]c.\label{princ_cont_val}
\end{equation}
This way of formulating the continuation value is intuitive: if the continuation value $V_{t}\left(\tau\right)$ gets negative then it is not valuable to continue investing for the principal, and he is better-off stopping immediately rather than following the schedule. The following Lemma shows the necessary and sufficient conditions for an investment schedule $\tau$ to generate obedient action recommendations for the principal.

\begin{lemma}\label{lemm_obedience} An investment schedule $\tau$  is the principal's best response to at least one information policy $\sigma$ if and only if
\begin{equation}
V_{t}\left(\tau\right)\geq 0,\forall t\geq0 \text{ and } V_{\tau}^{NI}<0, \label{obed_constr}
\end{equation}
where $V_{t}^{NI}$ is the principal's optimal continuation value in the absence of any additional information from the agent starting from the date $t$.
\end{lemma}
$V_{t}\left(\tau\right)\geq 0,\forall t\geq0$ ensures that the principal does not want to stop before the date of stopping suggested by the investment schedule is reached, and $V_{\tau}^{NI}<0$ ensures that the principal does not want to continue conditional on reaching the date of stopping suggested by the investment schedule. Conditions from Lemma \ref{lemm_obedience} constitute the system of constraints for the agent's problem.

As the agent chooses an investment schedule $\tau$ to maximize her long-run payoff, the constraint $V_{\tau}^{NI}<0$ is inactive at optimum.\footnote{Otherwise, the agent can prolong the expected funding by choosing a different $\tau$.} 
Thus, without loss of generality, I omit this constraint from the agent's problem, and the problem that the agent solves at $t=0$ is given by 
\begin{equation}
\begin{aligned} & \max_{\tau\in\mathcal{T}}\left\{ c\cdot\Esp\left[\tau\right]\right\} \\
 & \text{s.t.}\,V_{t}\left(\tau\right)\geq 0,\forall t\geq0,
\end{aligned}
\label{eq:OP_stop_time}
\end{equation}
where $\mathcal{T}$ is the set of stopping times with respect to
the natural filtration of $x_{t}$.
As the principal's choice to postpone the stopping of funding is costly, it is natural to interpret the system of constraints in (\ref{eq:OP_stop_time}) as the \emph{system of principal's individual rationality constraints}.

The agent's problem is complex, and thus I solve it in steps. First, I characterize  the investment schedule, which solves the \emph{relaxed} version of (\ref{eq:OP_stop_time}) with the principal's individual rationality constraints only for some initial periods. Second, I demonstrate that there exists an investment schedule solving the \emph{relaxed} agent's problem and satisfying the \emph{full} system of the principal's individual rationality constraints (\ref{obed_constr}). This investment schedule pins down optimal information policy.

\subsubsection{Solution to the relaxed agent's problem}\label{sect:relaxed}
In this Section, I consider the relaxed agent's problem and discuss its solution. This sheds light on the technical intuition behind the key properties o the optimal information policy.  The agent's relaxed problem for the parametric case of $\kappa\in(\kappa^{ND}\left(T,\lambda\right),\kappa^{NI}\left(T,\lambda\right)]$ is given by (\ref{eq:OP_stop_time}) with the \emph{principal's individual rationality constraint only for $t \in [0,\bar{S}^{NI}]$}. The agent's relaxed problem for the parametric case of $\kappa\in(\kappa^{NI}\left(T,\lambda\right),\tilde{\kappa}\left(T,\lambda\right)]$ is given by (\ref{eq:OP_stop_time}) with the \emph{principal's individual rationality constraint only for $t=0$}.

Consider the agent's long-run payoff given an investment schedule, $W\left(\tau\right)$.
This can be restated equivalently as follows:
\begin{equation}
\begin{aligned}
W\left(\tau\right) & = \left[W\left(\tau\right)+V\left(\tau\right)\right]-V\left(\tau\right)\\ & = 
\underbrace{\p\left(x_{\tau}=2\right)v}_{\text{total surplus}}\,-\underbrace{\left[\p\left(x_{\tau}=2\right)v-\Esp\left[\tau\right]c\right]}_{\text{principal's surplus}}.
\end{aligned}\label{zero_sum_pyf}
\end{equation}
The solution to the agent's relaxed problem for both considered parametric cases follows a simple idea: the optimal investment schedule \emph{ensures that the total surplus is maximal and that the principal's surplus is minimal}. 
Consider a schedule $\tau$ such that the stopping occurs after the
completion of the second stage of the project, unless the project
deadline $T$ was hit, i.e., the schedule satisfies the condition $\tau\geq\tau_{2}\land T$. Such a schedule leads to 
\begin{equation}
\p\left(x_{\tau}=2\right)=\p\left(x_{T}=2\right).\label{efficiency_cond}
\end{equation}

Given a schedule $\tau$ satisfying (\ref{efficiency_cond}), if $\tau$ is individually rational for the principal at date $t=0$ then the
total surplus generated achieves its upper bound and is given by $\p\left(x_{T}=2\right)v$,
which depends on the exogenously given project deadline $T$ and the profit $v$.
However, the stopping only after the second stage completion is not individually
rational for the principal at $t=0$ when the cost of funding is sufficiently
high, the profit is sufficiently low, or the expected time
until a project stage completion is sufficiently high. 

Lemma \ref{kappa_tilde} elaborates on the   cost-benefit ratio cutoff value $\tilde{\kappa}\left(T,\lambda\right)$: it distinguishes the
case in which stopping only after the second stage completion 
is individually rational at $t=0$ from the case in which it is not. Based on this
partition, when $\kappa\in(\kappa^{ND},\tilde{\kappa}\left(T,\lambda\right)]$,
I call the project \emph{ex ante promising} for the principal.

\begin{lemma}\label{kappa_tilde} For each $\left(T,\lambda\right)$
there exists $\tilde{\kappa}\left(T,\lambda\right)$, $\kappa^{NI}\left(T,\lambda\right)<\tilde{\kappa}\left(T,\lambda\right)<\kappa^{FI}\left(T,\lambda\right)$,
such that if $\kappa\leq\tilde{\kappa}\left(T,\lambda\right)$ $\left(\kappa>\tilde{\kappa}\left(T,\lambda\right)\right)$
then an investment schedule $\tau$ in which stopping after $\tau_{2}\land T$
happens with probability one is individually rational at $t=0$  (not individually
rational at $t=0$) for the principal. \end{lemma} 

For $\kappa\in(\kappa^{ND}\left(T,\lambda\right),\tilde{\kappa}\left(T,\lambda\right)]$, the schedule $\tau\geq\tau_{2}\land T$ is individually rational for the principal at $t=0$, and it  maximizes the total surplus. In addition to choosing $\tau\geq\tau_{2}\land T$, it is optimal for the agent to choose the investment schedule with a higher expected date of stopping the funding to extract all the principal's surplus subject to his individual rationality constraints.
For $\kappa\in(\kappa^{NI}\left(T,\lambda\right),\tilde{\kappa}\left(T,\lambda\right)]$, the agent chooses such $\tau$ that the principal's individual rationality constraint at $t=0$ is binding.  As a result, $V(\tau)=V^{NI}$, i.e., the principal gets his no-information benchmark payoff given by $0$.  

For $\kappa\in(\kappa^{ND}\left(T,\lambda\right),\kappa^{NI}\left(T,\lambda\right)]$, as in the no-information benchmark the principal invests until $\bar{S}^{NI}$ with certainty, the agent chooses the investment schedule as to postpone the start of information provision at least until $\bar{S}^{NI}$. Further, the agent chooses $\tau$ with a higher expected date of stopping so that the principal's individual rationality constraint at  $t=\bar{S}^{NI}$ is binding. The absence of stopping until at least $\bar{S}^{NI}$ and the fact that individual rationality constraint binds at $t=\bar{S}^{NI}$ taken together imply that $V(\tau)=V^{NI}$, i.e., from $t=0$ perspective, the principal gets her no-information benchmark payoff, which is non-negative and given by (\ref{NI_integral}). 

The next Lemma summarizes the  \emph{necessary}   conditions for an investment schedule to  solve the agent's relaxed problem when the project is promising. These conditions are shared both by the relaxed problem formulated for the case of  $\kappa\in(\kappa^{ND}\left(T,\lambda\right),\kappa^{NI}\left(T,\lambda\right)]$  and the relaxed problem formulated for the case of $\kappa\in(\kappa^{NI}\left(T,\lambda\right),\tilde{\kappa}\left(T,\lambda\right)]$. The conditions that are \emph{both necessary and sufficient} for an investment schedule to  solve the agent's relaxed problem are presented in the Proof of Lemma \ref{lemma_opt_sch1}.

\begin{lemma}\label{lemma_opt_sch1} Assume $\kappa\in(\kappa^{ND},\tilde{\kappa}\left(T,\lambda\right)]$.
If an investment schedule $\tau$ solves agent's relaxed problem, then 
\begin{itemize}
\item[1.] with probability one, stopping occurs after $\tau_{2}\land T$; 
\item[2.] $V\left(\tau\right)=V^{NI}$, where $V^{NI}$
is the principal's expected payoff in the no-information benchmark, given by (\ref{NI_integral}). 
\end{itemize}
\end{lemma}

\subsubsection{Optimal information policy}
In this Section, I show that there exists an information policy that both solves the agent's relaxed problem and satisfies the full system of the individual rationality constraints. Given this,  as Lemma \ref{lemma_opt_sch1} describes the solution to the relaxed problem, it also sheds light on the properties of the optimal information policy for the case of a promising
project. These properties have a clear-cut economic interpretation as an investment schedule $\tau$ can be easily interpreted in terms of action recommendations for the principal.

An investment schedule $\tau$ can be without loss of generality implemented using a \emph{direct
recommendation mechanism} - a simple policy which has $M=\left\{ 0,1\right\} $,
where $m=1$ received at date $t$ is a recommendation to continue
investing at $t$ for the principal and $m=0$ received at date $t$
is a recommendation to stop investing at $t$.\footnote{The connection
between an investment schedule $\tau$ and a direct recommendation
mechanism implementing the schedule $\tau$ is simple: whenever, based on the evolution of the state process, $\tau$ suggests stopping the funding, the direct recommendation mechanism sends the message $m=0$.}
Keeping this in mind, it is clear from  Lemma \ref{lemma_opt_sch1} that the \emph{optimal information policy} has to satisfy the following conditions. \emph{First}, whenever the agent recommends the principal to stop, the second stage
of the project is already completed. \emph{Second}, the recommendation to stop is postponed so that the principal's individual
rationality constraint is binding, which manifests in $V\left(\tau\right)=V^{NI}$.
The first condition presents the key feature of the optimal information policy for the case of promising project: the agent discloses the completion of the second stage of the project, but \emph{stays silent regarding the completion of the first stage of the project}. The intuition behind
the agent's choice is simple: a recommendation to stop when no stages
of the project are completed and the project deadline $T$ is close
does indeed incentivize the principal; however, it also reduces the total
surplus generated that can be extracted via the agent's control of
information. Meanwhile, the recommendation to stop when the two stages
of the project are completed incentivizes the principal without reducing
the total surplus generated. When $\kappa\leq\tilde{\kappa}\left(T,\lambda\right)$, a partially informative
policy that discloses only the completion of the second stage provides
sufficient incentives to the principal, and thus the agent uses it.\footnote{The ``leading on'' information policy in \citet{ely2020moving}
is similar: the only information that the policy provides is that
the task is already completed and, thus, it is time to stop investing.}

I proceed with obtaining  an investment schedule that not only satisfies the  conditions  in Lemma \ref{lemma_opt_sch1} and solves the relaxed problem, but also satisfies the full system of the principal's individual rationality constraints in Lemma \ref{lemm_obedience}. Ensuring both  is non-trivial. For instance, consider a mechanism that implements an investment schedule solving the agent's relaxed problem and assume it recommends to continue
for $t\in\left[0,S^{*}\right)$, then at  $S^{*}$
recommends stopping if the second stage is already completed, but
recommends to continue at all the subsequent dates $t\in(S^{*},T]$.
A no stopping recommendation drawn at $S^{*}$ reveals that the
state is either $0$ or $1$. Clearly, after sufficient time passes
after $S^{*}$, the principal would attach a high probability to
the second stage already being completed and would potentially be
tempted to deviate from the recommendation to continue.\footnote{In other words, $V_{t}\left(\tau\right)$ drifts down over time and can get negative at some date.} However, a \emph{direct recommendation mechanism that implements an optimal investment schedule exists}. I present it in Proposition
\ref{delayed_tau2}.

\begin{proposition}\label{delayed_tau2} Assume $\kappa\in(\kappa^{ND}\left(T,\lambda\right),\tilde{\kappa}\left(T,\lambda\right)]$. The optimal mechanism  does not provide
a recommendation to stop during $t\in[0,S^{*})$. At $t=S^{*}$,
if the second stage of the project is already completed, then the
mechanism recommends the principal to stop. If the second stage of the
project is not yet completed, then the mechanism recommends the principal
to stop at the moment of its completion $t=\tau_{2}$.
Formally, 
\[
\tau=S^{*}\lor\left(\tau_{2}\land T\right),
\]
where $S^{*}$ is chosen such that  $V\left(\tau\right)=V^{NI}$, i.e., the respective constraint in the system of principal's individual rationality constraints is binding.
\end{proposition}

The recommendation mechanism 
starting from $S^{*}$ generates recommendations to stop if the second stage is completed. As the recommendation to stop comes
immediately at the completion of the second stage for all $t>S^{*}$,
hearing no recommendation to stop reveals that the state is either
$0$ or $1$. Further, as time goes on, the principal attaches a higher
and higher probability to the state being $1$, which ensures obedience
to the recommendation to continue at each date. Further, the start of information
provision $S^{*}$ is sufficiently postponed to ensure that the  principal's individual rationality constraint is binding either at $t=\bar{S}^{NI}$ or at $t=0$.

The choice of $S^{*}$ is driven by extraction of the principal's surplus and depends on $\kappa$ in an intuitive way. First, consider the case $\kappa\in(\kappa^{ND},\kappa^{NI}\left(T,\lambda\right)]$,
the principal is willing to start investing and invests until $t=\bar{S}^{NI}$
in the no-information benchmark. 
The agent's optimal choice is to set $S^{*}>\bar{S}^{NI}$. Given such an information policy, the principal
does not stop at $\bar{S}^{NI}$, the date of stopping in the no-information
benchmark, and with probability one continues to invest during $t\in[\bar{S}^{NI},S^{*})$
even though the mechanism provides absolutely no information for all
$t<S^{*}$. This is driven by the fact that the expected benefit from
stopping at some future date $t\in[S^{*},T]$ and obtaining the project
payoff $v$ with certainty compensates the flow losses of investing
during $t\in[\bar{S}^{NI},S^{*})$.\footnote{Similarly to the ``leading on'' information policy in \citet{ely2020moving},
the promises of future disclosure of the completion of the project
are used as a ``carrot'' to make the receiver continue investing
beyond the point at which he stops in the no-information benchmark.}
Further, the agent sufficiently postpones $S^{*}$ to ensure that she extracts the principal's surplus and the principal gets precisely $V^{NI}\geq0$. 

In the case $\kappa\in(\kappa^{NI}\left(T,\lambda\right),\tilde{\kappa}\left(T,\lambda\right)]$, the principal is not willing to start in the no-information benchmark
as his expected payoff from investing is negative. Thus, the agent chooses $S^{*}$  to guarantee that the principal gets his reservation value  $V^{NI}=0$ and is thus willing to start investing at $t=0$. The value of $S^{*}$ is relatively lower as compared to the previous case: as the project is less attractive, to provide the principal sufficient incentives, the agent needs to start the information provision regarding the completion of the project earlier.

Finally, there exist many information policies that both solve the relaxed agent's problem and satisfy the full system of constraints (\ref{obed_constr}). This constitutes an important advantage for the agent: she can choose an optimal policy that is easier to implement from
the real-world perspective, depending on the particular environment.
In the optimal mechanism from Proposition \ref{delayed_tau2}, the recommendation
to stop at some date $t$ depends only on the current state of the
world $x_{t}$. In an alternative delayed disclosure mechanism, the recommendation to
stop arrives with a pre-specified delay after the second stage was
completed. Thus, the recommendation depends only on the past history
and not on the current state of the world. In an optimal delayed disclosure
mechanism, the delay becomes smaller as more time passes. I characterize
such a mechanism in Appendix \ref{app_delay}.\footnote{The rich set of optimal direct recommendation mechanisms in my model
encompasses both mechanisms in which the information provision depends
only on the current state, similarly to the optimal mechanism in \citet{ely2020moving},
and the mechanisms that use delay, similarly to the delayed beep from
\citet{ely2017beeps}.}

Recall that, as Lemma \ref{lemma_opt_sch1} suggests, the key idea of the optimal information policy
is that the agent postpones the disclosure of the completion of the
project to extract more surplus, which makes the principal's individual
rationality constraint bind. The higher the cost-benefit ratio of the project $\kappa$
becomes, the higher additional value the agent's information policy 
needs to provide to the principal to ensure that his active individual rationality constraint is satisfied. The implication of this
for the optimal information policy is presented in Lemma \ref{delay}.

\begin{lemma}\label{delay} Assume $\kappa\in(\kappa^{ND}\left(T,\lambda\right),\tilde{\kappa}\left(T,\lambda\right)]$.
Given the recommendation mechanism implementing an optimal investment
schedule $\tau$, for a fixed Poisson rate $\lambda$, the expected
length of investment $\Esp\left[\tau\right]$ decreases in the cost-benefit
ratio $\kappa$. \end{lemma}

The intuition is that the higher the cost-benefit ratio of the project
becomes, the sooner after the second stage of the project is completed
the agent recommends the principal to stop. For the cost-benefit ratio
as high as $\tilde{\kappa}\left(T,\lambda\right)$, the agent provides
the recommendation to stop immediately at the date of completion of
the second stage of the project. Further, for $\kappa>\tilde{\kappa}\left(T,\lambda\right)$,
the optimal information policy satisfying the conditions in Lemma \ref{lemma_opt_sch1} ceases
to be individually rational for the principal. As I show in the next Section, for $\kappa>\tilde{\kappa}\left(T,\lambda\right)$,
in addition to immediate disclosure of the project completion, the
agent provides the information regarding the completion of the first stage of the project.

\subsection{Immediate disclosure of completion and an interim deadline}

\label{section_INT}

When $\kappa>\tilde{\kappa}\left(T,\lambda\right)$, the project is
not promising for the principal and any investment schedule in which
stopping occurs after $\tau_{2}\land T$ with probability one violates
the principal's individual rationality constraint. In other words, from the ex ante perspective the
future reports disclosing only the completion of the project do not
motivate the principal to start investing. Thus, an investment schedule
that provides an individually rational expected payoff to the principal
should assign a positive probability not only to stopping after the
completion of the project, but also to stopping in either state $0$,
when no stages of the project are completed, or state $1$, when only
the first stage of the project is completed. I present the necessary
conditions for an investment schedule to be optimal 
when the project is not promising in Lemma \ref{lemma_effic_sch2}.

\begin{lemma}\label{lemma_effic_sch2} Assume $\kappa\in(\tilde{\kappa}\left(T,\lambda\right),\kappa^{FI}\left(T,\lambda\right))$.
If an investment schedule $\tau$ solves agent's problem, then it satisfies the conditions 
\begin{itemize}
\item[1.] conditional on no completed stages of the project, stopping of the funding happens
with a positive probability;
\item[2.] conditional on one completed stage of the project, stopping of the funding happens
with probability zero;
\item[3.] conditional on two completed stages of the project, stopping of the funding happens
immediately (at $t=\tau_{2}$) and with probability one.
\end{itemize}
\end{lemma}

Stopping when only the first stage of the project is already completed
is clearly inefficient. In state $1$, the principal prefers to continue
investing until the completion of the second stage and this principal's
incentive to wait is aligned with the agent's incentive to postpone
the stopping. Further, stopping in state $1$ does not help overcome
the problem of the violated individual rationality constraint under
$\kappa>\tilde{\kappa}\left(T,\lambda\right)$. Meanwhile, assigning
a positive probability to stopping when no stages are completed helps
to overcome the problem of violated individual rationality constraint, as the principal benefits from
stopping at some date $t$ when the first stage of the project is
not yet completed and the project deadline $T$ is sufficiently close.
Further, the agent clearly prefers the stopping of funding after the completion of
the second stage rather than in state $0$ as the former does not
harm the total surplus generated. Thus, a schedule that is optimal assigns probability $1$ to immediate
stopping when the second stage is completed. 

Lemma \ref{lemma_effic_sch2} implies that in an investment schedule, optimal for the agent, stopping after the completion of the second stage of the project happens immediately and stopping also happens given that no stages of the project are completed - i.e., at the \emph{interim deadline chosen by the agent, which I denote by} $S_{0}^{A}$. Thus, Lemma \ref{lemma_effic_sch2} drastically simplifies the strategy space in the agent's design problem: it is only left to characterize the optimal interim deadline.
At the outset of the game, the agent designs a device that privately randomizes over the interim deadlines $S_{0}^{A}$. That is, the agent publicly chooses a distribution $F_{S_{0}^{A}}$, then an interim deadline is drawn according to it and privately observed by the agent. Next, the information starts to flow. The action recommendation to stop the funding satisfies the following investment schedule
\begin{equation}
\tau=\begin{cases}
S_{0}^{A}, & \text{if }x_{S_{0}^{A}}=0\\
\tau_{2}\land T, & \text{otherwise},
\end{cases}
\end{equation}
where the principal knows only the distribution $F_{S_{0}^{A}}$, but not the draw of $S_{0}^{A}$.

Given that the completion of the second stage of the project is disclosed
immediately, stopping at the interim deadline in state $0$ leads
to a loss of expected further investment flow for the agent, and a
potential savings from abandoning a ``stagnating'' project for the principal.
The agent's payoff can be without loss of generality restated as the expected loss in future investment due
to stopping at the interim deadline $S_{0}^{A}$ in state $0$ (rather than at 
$\tau_{2}\land T$). Given this, the agent's problem can be expressed as 
\begin{equation}
\min_{F_{S_{0}^{A}}} \Esp_{F_{S_{0}^{A}}}\Bigg[\underbrace{\p\left(x_{S_{0}^{A}}=0\right)\Esp\left[\tau_{2}\land T-S_{0}^{A}|x_{S_{0}^{A}}=0\right]}_{\text{expected loss in future investment given }S_{0}^{A}}\Bigg], \label{S's_pyf_int_dealine}
\end{equation}
subject to the system of the principal's individual rationality constraints, which also have a natural interpretation as the expectation of principal's savings on the future investment, which discontinues at $S_{0}^{A}$ in state $0$, minus the loss in the project completion profit due to stopping the funding at $S_{0}^{A}$ in state $0$.\footnote{The principal's long-run payoff is given in (\ref{CaseC_constr}).}

Inspecting the agent's expected loss in future investment in (\ref{S's_pyf_int_dealine}) reveals that if the agent postpones the
interim deadline $S_{0}^{A}$, then two effects arise. First, the probability that
stopping at the interim deadline will happen decreases. Second, the
expected loss in total surplus due to stopping at the interim deadline
rather than at $\tau_{2}\land T$ decreases.
Thus, the agent's expected loss  in future investment  is decreasing
in the date of interim deadline and the agent prefers
an interim deadline with a later expected date.

Agent's choice of the interim deadline distribution $F_{S_{0}^{A}}$ is affected
by the two factors. First, as the expected loss in future investment in (\ref{S's_pyf_int_dealine}) is
decreasing and convex in the date of the interim deadline, and thus the agent is risk-averse with respect to random interim deadlines.
Thus, given some random interim deadline, the agent \emph{directly}
benefits from inducing a mean-preserving contraction. Second, the
agent benefits from inducing a mean-preserving contraction \emph{indirectly}.
Inspecting the principal's long-run payoff for some fixed $S_{0}^{A}$ reveals that the principal
is also risk-averse with respect to random interim deadlines. Thus,
inducing a mean-preserving contraction makes the principal better-off
and relaxes the individual rationality constraint at $t=0$, hence, allowing
the agent to postpone the expected interim deadline further. As a
result the optimal for the agent interim deadline takes the form of
a \emph{deterministic date}. In other words, it is optimal for the agent to \emph{publicly announce the interim deadline $S_{0}^{A}$ at the outset}, so that the principal knows it. 


The agent has an incentive to postpone the interim
deadline and uses her control of the information environment to postpone
the deadline as much as possible so that the principal's individual
rationality constraint at $t=0$ binds. Figure \ref{fig_tauint} demonstrates
the principal's long-run payoff as a function of the interim
deadline, which I denote by $S_{0}$. It is maximized at the principal-preferred interim
deadline $S_{0}^{P}$, which was characterized in Lemma \ref{lemma_FI}. The agent-preferred interim deadline $S_{0}^{A}$
yields the principal's expected payoff of $0$. 
\begin{figure}[H]
\centering \begin{tikzpicture}[xscale = 2,yscale = 2,domain=1:3.3,samples = 200]
\draw[thick,->] (0.9,1) node[below]{$0$} -- (4,1) node[below]{$S_{0}$}; 
\draw[thick,->] (1,0) -- (1,2.45) node[left]{$V$}; 
\draw[thick,red!50!black] plot (\x,{2-(\x-2)^2});
\draw[ultra thin,,dashed] (2,1) node[below]{$S_{0}^{P}$} -- (2,0.9);
\draw[ultra thin,,dashed] (3.5,1) node[below]{$T$} -- (3.5,0.9);
\draw[ultra thin,,dashed] (3,1) node[below]{$S_{0}^{A}$} -- (3,0.9);
\end{tikzpicture} \caption{Principal's long-run payoff, $V$, as a function
of an interim reporting deadline chosen by the agent, $S_{0}$.}
\label{fig_tauint} 
\end{figure}
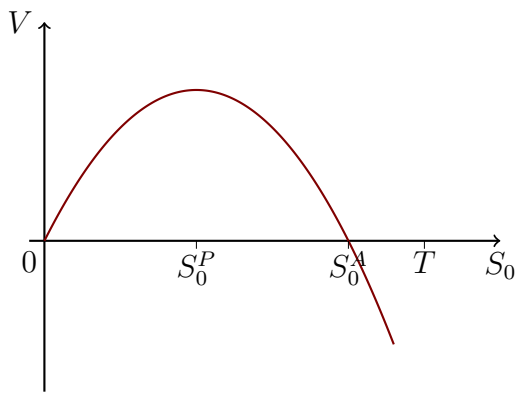

The next Proposition summarizes the optimal investment schedule, which can be without loss of generality implemented using a direct
recommendation mechanism: 

\begin{proposition}\label{case3} Assume $\kappa\in(\tilde{\kappa}\left(T,\lambda\right),\kappa^{FI}\left(T,\lambda\right))$.
The optimal information policy is given by a direct recommendation
mechanism that generates 
\begin{enumerate}
    \item[(a)] the recommendation to stop at the moment
of completion of the second stage of the project, $t=\tau_{2}$, and
\item[(b)] a conditional recommendation to stop at the publicly announced interim deadline $t=S_{0}^{A}$.
At $S_{0}^{A}$, stopping is recommended with certainty if the first
stage of the project has not yet been completed. 
\end{enumerate}
Formally,
\begin{equation*}
\tau=\begin{cases}
S_{0}^{A}, & \text{if }x_{S_{0}^{A}}=0\\
\tau_{2}\land T, & \text{otherwise},
\end{cases}
\end{equation*}
where $S_{0}^{A}$ is chosen
so that the principal's individual rationality constraint at $t=0$ is binding, i.e., $V\left(\tau\right)=0$.

\end{proposition}

A stopping recommendation at any date other than the interim deadline $t=S_{0}^{A}$ fully reveals that project is accomplished. Further, observing a recommendation to stop at the interim deadline, the principal learns that the milestone of the project
has not yet been reached and becomes sufficiently pessimistic that
the project will be completed by $T$.

A notable feature of the optimal information policy when
the project is ex ante unattractive is its uniqueness. The only
optimal instrument through which the agent fine tunes the incentive
provision to the principal is the choice of interim deadline, and
there is a unique optimal way to set the deadline to make the principal's
individual rationality constraint bind.

I proceed by considering the comparative statics of the interim deadline.
Both the agent-preferred and the principal-preferred interim deadline,
$S_{0}^{A}$ and $S_{0}^{P}$, respectively, increase in the exogenous
deadline $T$. This is because less time pressure relaxes the principal's
individual rationality constraint and allows the agent to postpone
the deadline further in order to extract the principal's surplus.

As the cost-benefit ratio increases up to $\kappa^{FI}$, the agent-preferred
deadline converges to the principal-preferred deadline. An increase
in the cost-benefit ratio of the project makes the principal's
individual rationality constraint tighter.\footnote{This is because the increase in $\kappa$ makes the principal's instantaneous
benefit from waiting decrease, and the normalized instantaneous cost
of waiting becomes higher.} As a result, for a higher $\kappa$, in the absence of completion
of the first stage, the principal is willing to wait for a shorter
time before stopping. Thus, both the interim deadline preferred by
the principal $S_{0}^{P}$ and the interim deadline chosen by the
agent $S_{0}^{A}$ are lower for a higher $\kappa$. Further, for a
higher $\kappa$ the agent has to choose an information policy relatively
closer to the full-information benchmark to ensure that the individual rationality constraint at $t=0$ is satisfied. Hence, the agent-chosen
interim deadline $S_{0}^{A}$ approaches $S_{0}^{P}$, the interim
deadline preferred by the principal. The comparative statics of $S_{0}^{P}$
and $S_{0}^{A}$ with respect to the cost-benefit ratio of the project
$\kappa$ are presented in Figure \ref{S_int_kappa_fig}.

\begin{figure}[H]
\captionsetup{justification=centering} \centering \includegraphics[width=0.45\textwidth]{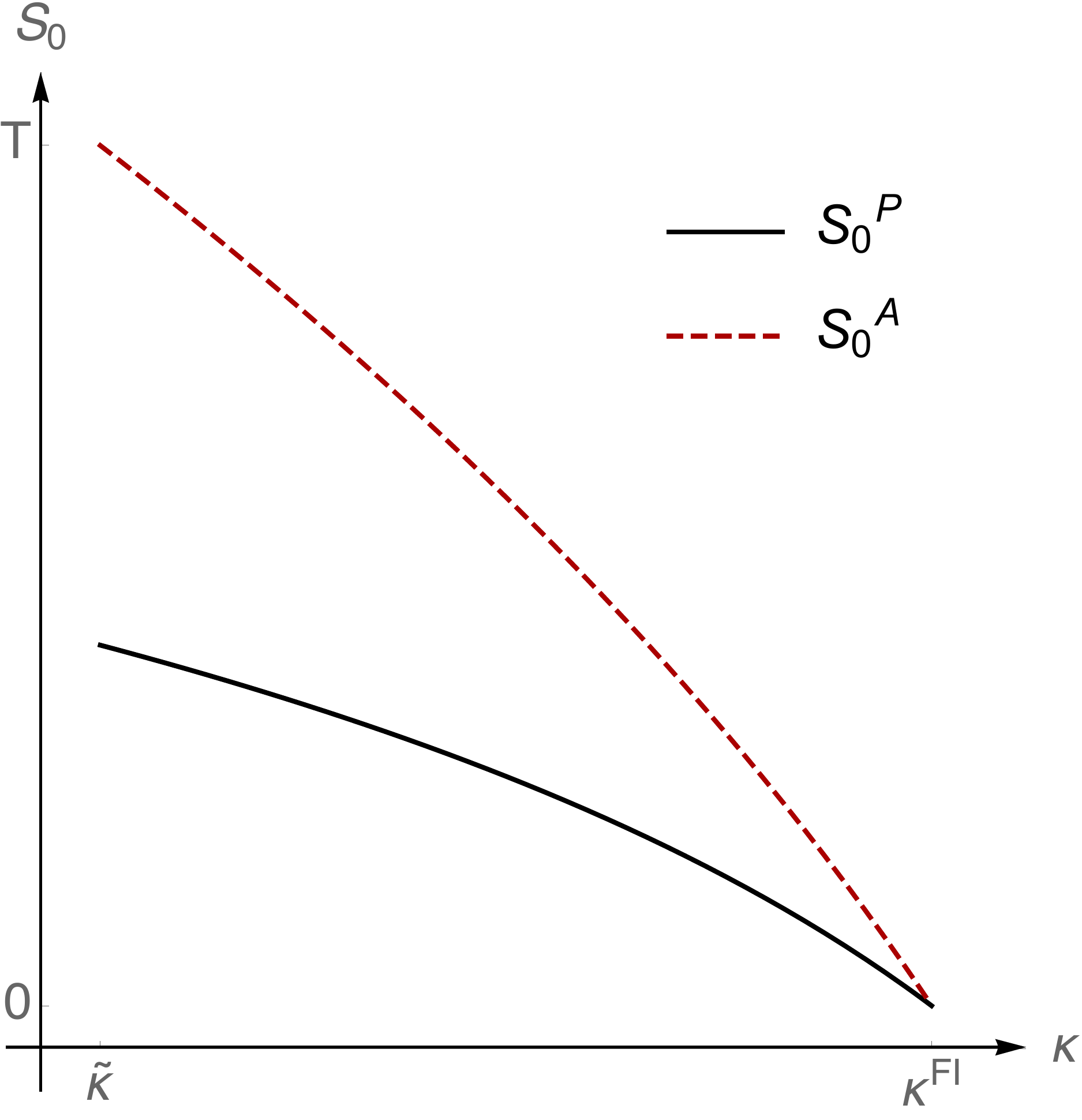}
\caption{Interim deadline chosen by the agent $S_{0}^{A}$ \textbf{(dashed)}
and preferred by the principal $S_{0}^{P}$ \textbf{(thick)}, as
functions of the cost-benefit ratio of the project $\kappa$.}
\label{S_int_kappa_fig} 
\end{figure}

\section{General preferences}\label{sec:prof_share}

In this Section, I allow for profit-sharing between the agent and
the principal, varying degree of the agent's benefit from the flow of funds, and exponential discounting, and demonstrate that the optimal
information policy still has the same properties as in the baseline
model.

First, I assume that the agent and the principal share the project
completion profit $v$: the principal gets $\alpha
\cdot v$, while the agent gets $(1-\alpha)
\cdot v$, $\alpha \in
\left(0,1\right]$. Thus, now the
agent benefits not only from the flow of funds provided by the principal
for running the project but also from the share in the profit. The
assumption that the agent gets a share in the project completion profit
is natural in many situations. In particular, the research documents
that the entrepreneurs in innovative startups are up to some extent
driven by giving vent to their entrepreneurial mindset and bringing
their innovative ideas to life \citep{gundolf2017startups}.
In such a context, a positive profit share of the agent captures that
the agent is motivated by the success of the project.

Second, I assume that given a flow cost of $c$ incurred by the principal,
the agent obtains a flow benefit $\beta c,\beta\geq0$.
$\beta$ can be interpreted as the agent's marginal benefit from using
the funds provided by the principal for funding the project. Alternatively, for $\beta\in\left[0,1\right]$
the loss of $1-\beta$ of the amount of the transfer at each date
can be interpreted as the transaction costs. Finally, setting $\beta=0$ for some $\alpha<1$
allows for abstracting from the agent's motives for diverting the
funds and considering the agent motivated only by the success of the
project.

Third, I allow for exponential discounting at a rate $r>0$. Thus,
the present value of a profit obtained at a date $t$ is given by
$ve^{-rt}$ and the present value of a stream of funding up to date
$t$ is given by $\frac{1}{r}\left(1-e^{-rt}\right)c$. The following
Proposition demonstrates that given the more general preference specification,
the structure of the optimal disclosure, present in the baseline model,
preserves.

\begin{proposition}\label{prop:prof_share1}
\mbox{}\\*
\vspace{-14pt}
\begin{itemize}
    \item[(a)] When the cost-benefit ratio of the project is low, $\kappa\leq\tilde{\kappa}\left(T,\lambda,r,\alpha\right)$, the optimal investment schedule $\tau$ satisfies $\tau\geq\tau_{2}\land T$, i.e., the agent recommends the principal to stop only after the completion of
the second stage of the project. 
\item[(b)] When $\kappa>\tilde{\kappa}\left(T,\lambda,r,\alpha\right)$,
the optimal investment schedule $\tau$ assigns positive probability both to the stopping in state $2$ and state $0$, i.e., the agent not only discloses the completion of
the second stage of the project, but also specifies an interim deadline
for the completion of the first stage.
\end{itemize}
\end{proposition}

Similarly to the baseline model, allowing the principal to stop after the project completion brings profit to the principal and thus leads to a relatively higher total surplus, which the agent can extract. Meanwhile, allowing the principal to stop at the interim deadline does not increase total surplus and serves solely as an expected payoff transfer from the agent to the principal. To see that, note that stopping when the first stage of the project is still incomplete allows the principal to save on the further costs of funding the project when over time the project proves to be ``unsuccessful''. This can not be beneficial for the agent as she does not internalize the costs of running the project. Further, stopping at the interim deadline is strictly detrimental for the agent as she strictly prefers the principal to postpone the stopping of funding when no stages of the project are completed.\footnote{The probability of project success and stock of obtained funds are non-decreasing in the date of stopping irrespective of the number of the completed stages of the project.}

When the project is sufficiently ex-ante attractive, the agent can
motivate the principal to start funding the project without promising
to stop the stagnant project at the interim deadline, and this is strictly beneficial for the agent. Thus, when the project is promising, the agent sets no interim deadlines,
which in expectation gives her more funds and more experimentation for free.

\section{Conclusion}

A transparent flow of information is crucial for the successful management of any innovative project. However, the researcher, who controls the information on the progress of the project, often tends to have different motives than the investor. This leads to the question of how a researcher chooses the transparency of the flow of information about the progress of a project in order to manipulate the investor's funding decisions. I address this question in a dynamic information design model in which the agent commits to providing information to the principal with an incentive to postpone the principal's irreversible stopping of the funding.

I contribute to the dynamic information design literature by studying the problem of the dynamic provision of information regarding the progress of a multistage project, which evolves endogenously over time and needs to be completed before a deadline. I show that the agent's choice of which pieces of information to provide and when depends on the project being either ex ante attractive for the principal or not. In the case of a promising project, the agent provides only the good news that the project is completed and postpones the reports. In the case of an unattractive project, to motivate the principal to start funding the project the agent not only reports the completion of the project, but also helps the principal to find out when the project stagnates. To achieve this, the agent announces an interim deadline for the project -- a certain date at which she recommends the principal to cut the funding of the project if the milestone of the project has not been reached.

\addcontentsline{toc}{section}{References} \bibliography{main}
 \bibliographystyle{apalike}

\appendix
\section*{Appendix}
\addcontentsline{toc}{section}{Appendices}
\renewcommand{\thesubsection}{\Alph{subsection}}

\subsection{Notational conventions}

The state process is given by $x_{t},\forall t\in\mathbb{R}_{+}$, defined on the probability
space $\left(\Omega,\mathcal{F},\mathbb{P}\right)$, $t\in\mathbb{R}_{+}$.
Its natural filtration is denoted by $F=(\mathcal{F}_{t})_{t\geq0}$.
Throughout Appendices \ref{app_val_funct} and \ref{app-proofs},
the following notational conventions are used:

1. I denote the random time at which the $n$th stage of the project
is completed by $\tau_{n}$. Formally, $\tau_{n}\in\mathbb{R}_{+}$
is a continuously distributed random variable that represents the
first hitting time of $x_{t}=n$.

2. Consider some stopping time $\tau$. Whenever ``$\tau$'' stands
as a term in an inequality, it stands for a realization of the stopping
time $\tau$ and it should be read as ``for each $\omega\in\Omega$
and corresponding $\tau\left(\omega\right)$''.

\emph{Example 1.} ``$\tau_{2}\land T\geq\tau$'' should be read
as ``$\tau_{2}\left(\omega\right)\land T\geq\tau\left(\omega\right)$,
for all $\omega\in\Omega$''.

\emph{Example 2.} ``for all $t\in[S,\tau)$'' should be read as
``for all $t\in[S,\tau\left(\omega\right))$, for all $\omega\in\Omega$''.

3. The continuation value of the agent at time $t$, given $\tau$,
and conditional on $t<\tau$: 
\[
W_{t}\left(\tau\right)\coloneqq\Esp\left[\tau-t|t<\tau\right]c.
\]

4. The total (continuation) surplus at time $t$, given $\tau$,
and conditional on $t<\tau$:
\[
SV_{t}\left(\tau\right)\coloneqq W_{t}\left(\tau\right)+V_{t}\left(\tau\right).
\]

5. Shorthand for posterior beliefs:

\[
\begin{array}{c}
q_{n}\left(t\right)\coloneqq\p\left(x_{t}=n|t<\tau\right),\\
r_{n}\left(t\right)\coloneqq\p\left(x_{\tau}=n|t<\tau\right).
\end{array}
\]

\subsection{The principal’s continuation value}

\label{app_val_funct}

Here I present the alternative formulation of the principal's continuation
value (\ref{princ_cont_val}). This helps me to study some of its
properties for further use in Appendix \ref{app-proofs}. The continuation
value of the principal at time $t$ and given the investment schedule
$\tau$ is given by (\ref{princ_cont_val}). Postponing the stopping
for $\Delta t$ brings a benefit in the form of project completion
payoff $v$  iff the second stage of the project is completed within
$\Delta t$. As $x_{t}$ follows the Poisson process, the probability
of two increments in a very short time $\Delta t$ is negligibly small.
Thus, during some $\Delta t$, the principal gets the project completion
payoff $v$ iff the first stage of the project has already been completed
at some earlier time. Thus, postponing the stopping for $\Delta t$
brings the principal $v$ with probability $\lambda q_{1}\left(t\right)\Delta t$.
The second stage is not completed within $\Delta t$ with the complementary
probability of $1-\lambda q_{1}\left(t\right)\Delta t$. The principal's
continuation value is thus given by 
\[
\begin{aligned}V_{t}\left(\tau\right) & =\left(v\lambda q_{1}\left(t\right)-c\right)\Delta t+\left(1-\lambda q_{1}\left(t\right)\Delta t\right)V_{t+\Delta t}\left(\tau\right)\\
 & =v\lambda\left(q_{1}\left(t\right)-\kappa\right)\Delta t+\left(1-\lambda q_{1}\left(t\right)\Delta t\right)V_{t+\Delta t}\left(\tau\right).
\end{aligned}
\]
Differentiating both sides w.r.t. $\Delta t$ and considering $\lim_{\Delta t\rightarrow0}\left(.\right)$
yields 
\[
0=v\lambda\left(q_{1}\left(t\right)-\kappa\right)-\lambda q_{1}\left(t\right)V_{t}\left(\tau\right)+\dot{V_{t}}\left(\tau\right),
\]
which, after rearranging becomes 
\begin{equation}
\dot{V_{t}}\left(\tau\right)=\lambda q_{1}\left(t\right)V_{t}\left(\tau\right)+v\lambda\left(\kappa-q_{1}\left(t\right)\right).\label{app_V_dot}
\end{equation}

\subsection{Proofs}

\label{app-proofs}

\begin{proof}[Proof of Lemma \ref{no_feedb}]

The beliefs regarding the number of stages of the project completed
by time $t$, $x_{t}$, evolve according to the Poisson process. The
principal's unconditional beliefs are given by $p_{0}\left(0\right)=1$
and for any $t$ such that the game still continues, 
\begin{eqnarray}
\dot{p}_{0}\left(t\right) & = & -\lambda p_{0}\left(t\right),\nonumber \\
\dot{p}_{1}\left(t\right) & = & \lambda(p_{0}\left(t\right)-p_{1}\left(t\right)),\label{prior_system}\\
\dot{p}_{2}\left(t\right) & = & \lambda p_{1}\left(t\right),\nonumber 
\end{eqnarray}
where $p_{0}\left(t\right)=e^{-\lambda t}$ and $p_{1}\left(t\right)=\lambda t e^{-\lambda t}$,
$p_{2}\left(t\right)=1-p_{0}\left(t\right)-p_{1}\left(t\right)$.
The principal's problem is given by 
\begin{equation}
\max_{S\in\left[0,T\right]}\left\{v\cdot p_{2}\left(S\right)-c\cdot S\right\}.\label{no_inf_OP_app}
\end{equation}

I start with analyzing the choice of $S$ for the interior solution
case, $S\in\left(0,T\right)$. F.O.C. for (\ref{no_inf_OP_app}) is
given by 
\begin{equation}
v\cdot\dot{p}_{2}\left(S\right)=c,\label{no_inf_FOC_app}
\end{equation}
or, equivalently, $p_{1}\left(S\right)=\kappa$. There are two values
satisfying (\ref{no_inf_FOC_app}): $\bar{S}$ and $\bar{S}^{NI}$, $\bar{S}<\bar{S}^{NI}$.
At each $t\in\left(\bar{S},\bar{S}^{NI}\right)$ the principal receives a net
positive payoff flow. Thus, stopping at $\bar{S}$ is not optimal and
the only candidate for optimal stopping is $\bar{S}^{NI}$.\footnote{$\bar{S}$ is a local minimum of the objective.}
Further, one can obtain the closed form expression for the interior
stopping time $\bar{S}^{NI}$ from (\ref{no_inf_FOC_app}): 
\begin{equation}
\bar{S}^{NI}=-\frac{1}{\lambda}\mathcal{W}_{-1}\left(-\kappa\right),\label{T_tilde_tilde}
\end{equation}
where $\mathcal{W}_{-1}(x)$ denotes the negative branch of the Lambert
$W$ function. $\bar{S}^{NI}$ is well-defined for any $\kappa<e^{-1}$.

Thus, the solution to (\ref{no_inf_OP_app}) could potentially be
$0$, $\bar{S}^{NI}$, or $T$. I proceed with a useful lemma.

\begin{lemma}\label{NI_obed} The following is true regarding the
principal's continuation value in the no-information benchmark, $\overline{V}_{t}^{NI}$:
if $\overline{V}_{t}^{NI}\geq0,\text{ for some }t\in\left[0,\bar{S}^{NI}\land T\right]$,
then $V^{NI}\left(s\right)\geq0,\text{ for all }s\in\left[t,\bar{S}^{NI}\land T\right]$.
\end{lemma}

\begin{proof} The principal's continuation value in the no-information
benchmark is given by 
\begin{equation}
\overline{V}_{t}^{NI}=\left[p_{2}\left(T\land \bar{S}^{NI}\right)-p_{2}\left(t\right)\right]v-\left(T\land \bar{S}^{NI}-t\right)c.\label{V_NI_integr}
\end{equation}
Further, 
\[
\dot{V}^{NI}\left(t\right)=v\lambda\left(\kappa-e^{-\lambda t}\lambda t\right)=v\lambda\left(\kappa-p_{1}\left(t\right)\right).
\]
$p_{1}\left(t\right)\leq\kappa$ for all $t\in\left[0,\bar{S}\right]$
and $p_{1}\left(t\right)\geq\kappa$ for all $t\in\left[\bar{S},\bar{S}^{NI}\land T\right]$.
Thus, $\overline{V}_{t}^{NI}$ increases for $t\in\left[0,\bar{S}\right]$,
decreases for $t\in\left[\bar{S},T\land \bar{S}^{NI}\right]$, and $V^{NI}\left(T\land \bar{S}^{NI}\right)=0$,
which establishes the result. \end{proof}

Lemma \ref{NI_obed} implies that if $V^{NI}\left(0\right)\geq0$
and the principal chooses to opt in at $t=0$, then $\overline{V}_{t}^{NI}\geq0$, 
$t\in\left[0,\bar{S}^{NI}\land T\right]$, i.e., he invests until $t=T\land \bar{S}^{NI}$.
This implies that the solution to (\ref{no_inf_OP_app}) is either
$T\land \bar{S}^{NI}$ or $0$.

Finally, at $t=0$ the principal chooses to start investing or not.
The condition for the principal to start investing at $t=0$ is given
by 
\begin{equation}
V^{NI}\geq0.\label{no_inf_partic_app}
\end{equation}
To specify the set of parameters for which (\ref{no_inf_partic_app})
is satisfied, I obtain the cutoff value of $\kappa$ given $T$ and
$\lambda$. Such a parameterization is intuitive: $\kappa$ above
the cutoff level corresponds to a project with sufficiently high normalized
cost-benefit ratio and implies that the principal does not opt in.
I denote this cutoff by $\kappa^{NI}\left(T,\lambda\right)$. This
solves (\ref{no_inf_partic_app}) holding with equality. Two cases
are possible.

\emph{Case 1}: $T\leq \bar{S}^{NI}\iff T\leq-\frac{1}{\lambda}\mathcal{W}_{-1}\left(-\kappa\right)$.
This inequality is satisfied when either $\frac{1}{\lambda}> T$
or $\begin{cases}
\frac{1}{\lambda} & \leq T\\
\kappa & \leq e^{-\lambda T}\lambda T.
\end{cases}$ Given $T\leq \bar{S}^{NI}$, (\ref{no_inf_partic_app}) holding with equality
becomes 
\[
p_{2}\left(T\right)v-Tc=0.
\]
Solving it for $\kappa$ yields $\kappa=e^{-\lambda T}\left(\frac{e^{\lambda T}-1}{\lambda T}-1\right)$.

\emph{Case 2:} $T>\bar{S}^{NI}$. This inequality is satisfied when $\frac{1}{\lambda}\leq T\text{ and }\kappa> e^{-\lambda T}\lambda T.$
Given $T>\bar{S}^{NI}$, (\ref{no_inf_partic_app}) holding with equality
becomes

\[
vp_{2}\left(\bar{S}^{NI}\right)-c\bar{S}^{NI}=0\iff v\left(1-p_{0}\left(\bar{S}^{NI}\right)-p_{1}\left(\bar{S}^{NI}\right)\right)=c\bar{S}^{NI},
\]
where (recall that $\dot{p}_{2}\left(\bar{S}^{NI}\right)=\frac{c}{v}$) 
\[
p_{0}\left(\bar{S}^{NI}\right)=\frac{1}{\lambda^{2}\bar{S}^{NI}}\dot{p}_{2}\left(\bar{S}^{NI}\right)=\frac{c}{\lambda^{2}\bar{S}^{NI}v}=\frac{\kappa}{\lambda \bar{S}^{NI}}
\]
and 
\[
p_{1}\left(\bar{S}^{NI}\right)=\frac{1}{\lambda}\dot{p}_{2}\left(\bar{S}^{NI}\right)=\frac{c}{\lambda v}=\kappa.
\]
Consequently, 
\[
vp_{2}\left(\bar{S}^{NI}\right)-c\bar{S}^{NI}=v-v\cdot\kappa\left(1+\lambda \bar{S}^{NI}+\frac{1}{\lambda \bar{S}^{NI}}\right).
\]
Let $y\coloneqq\lambda \bar{S}^{NI}$. Note that, by definition, $y>1$.
Then $\kappa=ye^{-y}$, and so 
\[
\left(vp_{2}\left(\bar{S}^{NI}\right)-c\bar{S}^{NI}\right)/v=1-e^{-y}\left(1+y+y^{2}\right).
\]
It follows that $V^{NI}\left(0\right)$ is nonnegative whenever $\lambda \bar{S}^{NI}\geq y_{0}\doteq1.79328$,
which is equivalent to 
\[
\kappa\leq\kappa_{0}\doteq0.298426.
\]

Finally, putting the two cases together yields

\begin{equation}
\kappa^{NI}\left(T,\lambda\right)=\begin{cases}
\kappa_{0}\doteq0.298426, & \text{if }\frac{1}{\lambda}\leq T\text{ and }\kappa\geq e^{-\lambda T}\lambda T\\
e^{-\lambda T}\left(\frac{e^{\lambda T}-1}{\lambda T}-1\right), & \text{otherwise.}
\end{cases}\label{kappa_NI}
\end{equation}

\end{proof}

\begin{proof}[Proof of Lemma \ref{lemma_FI}]

The principal chooses $a_{t}\in\left\{ 0,1\right\} $ sequentially
given the observed realizations of $x_{t}\in\left\{ 0,1,2\right\} $.
Whenever the principal observes $t=\tau_{2}$, he immediately chooses
$a_{t}=0$ and gets $v$.

Consider the case $x_{t}=1,t<T$, i.e., the first stage of the project
has already been completed. As $x_{t}$ follows a Poisson process,
in expectation it would take $\frac{1}{\lambda}$ units of time for
the second stage to be completed and its completion brings the principal
the value of $v$. Thus, the necessary and sufficient condition for
the principal to invest at $t$ when $x_{t}=1,t<T$ is given by 
\[
v-c\cdot\frac{1}{\lambda}\geq0\iff\kappa\leq1
\]
Assume that $\kappa\leq1$ holds and $x_{t}=1$; thus, the principal
chooses to invest at $t$. In that case, the principal invests until
$\tau_{2}\land T$. To see this, recall that the only news that the
principal can receive given $x_{t}=1,t<T$ is the completion of the
second stage of the project, $\tau_{2}$, which leads to immediate
stopping. At each $t<\tau_{2}\land T$ such that $x_{t}=1$, choosing
$a_{t}=0$ yields an instantaneous expected payoff of $0$, while
choosing $a_{t}=1$ yields an instantaneous expected payoff of $\lambda v\Delta t-c\Delta t$.
Thus, $\kappa\leq1$ suffices for the principal to invest until $\tau_{2}\land T$.

Consider now the case of $x_{t}=0,t<T$, i.e., no stages of the project
have yet been completed. Postponing the stopping for $\Delta t$ brings
the instantaneous expected payoff of $V_{t|1}^{FI}\lambda\Delta t-c\Delta t$,
where $V_{t|1}^{FI}$ is the principal's continuation
value at time $t$ under full information, conditional on the completion
of the first stage of the project. I proceed with obtaining the expression
for $V_{t|1}^{FI}$. By definition, the principal gets
$v$ whenever the second stage is completed not later than $T$. The
principal invests until $\tau_{2}\land T$, and knows that at $t$
the first stage of the project is already completed; thus, $V_{t|1}^{FI}$
is given by

\[
V_{t|1}^{FI}=v\p\left(\tau_{2}\leq T|x_{t}=1\right)-c\Esp\left[\tau_{2}\wedge T-t|x_{t}=1\right].
\]
$\tau_{2}|x_{t}=1$ corresponds to the time between two consecutive
Poisson arrivals, and thus has exponential distribution. First, consider
$\p\left(\tau_{2}\leq T|x_{t}=1\right)$: 
\[
\p\left(\tau_{2}\leq T|x_{t}=1\right)=1-e^{-\lambda\left(T-t\right)}.
\]
Next, consider $\Esp\left[\tau_{2}\wedge T-t|x_{t}=1\right]$: 
\begin{equation}
\begin{aligned} & \Esp\left[\tau_{2}\land T|x_{t}=1\right]-t\\
= & \p\left(\tau_{2}\leq T|x_{t}=1\right)\int_{t}^{T}z\cdot\frac{\lambda e^{-\lambda\left(z-t\right)}}{\p\left(\tau_{2}\leq T|x_{t}=1\right)}dz+\p\left(\tau_{2}>T|x_{t}=1\right)T-t\\
= & \frac{1}{\lambda}\left(1-e^{-\lambda\left(T-t\right)}\right)+t-e^{-\lambda\left(T-t\right)}T+\p\left(\tau_{2}>T|x_{t}=1\right)T-t\\
= & \frac{1}{\lambda}\left(1-e^{-\lambda\left(T-t\right)}\right).\label{cond_{e}xp_{t}aux1}
\end{aligned}
\end{equation}
Thus,

\begin{equation}
\begin{aligned}V_{t|1}^{FI} & =v\left(1-e^{-\lambda\left(T-t\right)}\right)-c\frac{1}{\lambda}\left(1-e^{-\lambda\left(T-t\right)}\right)\\
 & =\left(v-\frac{c}{\lambda}\right)\left(1-e^{-\lambda\left(T-t\right)}\right).
\end{aligned}
\label{V_1_FI_app}
\end{equation}
From (\ref{V_1_FI_app}) one observes that $V_{t|1}^{FI}$
decreases in $t$. If the net instantaneous benefit given by $V_{t|1}^{FI}\lambda\Delta t-c\Delta t$
gets as low as $0$ at some $t$, then the principal chooses to stop
investing at this $t$. I denote the time at which the net instantaneous
benefit reaches $0$ by $S_{0}^{P}$. $S_{0}^{P}$ can be obtained
from $\left(\lambda V_{1}^{FI}\left(S_{0}^{P}\right)-c\right)\Delta t=0$.
Thus, 
\begin{equation}
S_{0}^{P}=T+\frac{1}{\lambda}\log\left(\frac{1-2\kappa}{1-\kappa}\right).\label{T_0_FI}
\end{equation}

The principal is willing to start investing iff at $t=0$ the expected
payoff from investing at $t=0$ covers the costs of investing, i.e.
$\left(\lambda V_{1}^{FI}\left(0\right)-c\right)\Delta t\geq0$. From
(\ref{T_0_FI}), this corresponds to $S_{0}^{P}\geq0$. I denote
the upper bound on the cost-benefit ratio $\kappa$ such that the
principal chooses to start investing in $t=0$ under full information
by $\kappa^{FI}\left(T,\lambda\right)$, I solve $S_{0}^{P}=0$ for
$\kappa$ and obtain 
\begin{equation}
\kappa^{FI}\left(T,\lambda\right)=\frac{1-e^{-\lambda T}}{2-e^{-\lambda T}}.\label{kappa_FI}
\end{equation}

In summary, under full information, if $\kappa\leq\kappa^{FI}\left(T,\lambda\right)$,
then the principal starts investing at $t=0$. Further, he stops at
$S_{0}^{P}$ if the first stage of the project has not been completed
by that time. Otherwise, he proceeds to invest until $\tau_{2}\land T$.

\end{proof}

\begin{proof}[Proof of Proposition \ref{prop:opt_discl}] I provide the proof for each of the four parametric cases below.

\emph{1. The case of $\kappa\leq\kappa^{ND}\left(T,\lambda\right)$}. 

$\kappa^{ND}\left(T,\lambda\right)$ is defined as follows: for any $\kappa\leq\kappa^{ND}\left(T,\lambda\right)$,
the principal invests until $T$ in the no-information benchmark.
From Lemma \ref{no_feedb}, if the principal is willing to start
investing, i.e., $\kappa\leq\kappa^{NI}\left(T,\lambda\right)$, then
\[
S^{NI}=\overline{S}^{NI}\land T.
\]
For the sake of instruction, below I consider relaxing the Assumption \ref{assumpt1}
and demonstrate how the relation between $\kappa^{ND}\left(T,\lambda\right)$
and $\kappa^{NI}\left(T,\lambda\right)$ changes between Case a (relaxed
Assumption \ref{assumpt1}) and Case b (Assumption \ref{assumpt1} holds). 

Case a. $e^{\lambda T}\leq\lambda T\left(\lambda T+1\right)+1$. In
this case, whenever the principal is willing to start investing in
the no-information benchmark, she invests until $T$, i.e., $\kappa^{ND}\left(T,\lambda\right)=\kappa^{NI}\left(T,\lambda\right)$,
where $\kappa^{NI}\left(T,\lambda\right)$ is given
by (\ref{kappa_NI}).  To see that, first, consider the extreme sub-case in which $T<\frac{1}{\lambda}$.
As $-\lambda\overline{S}^{NI}$ must belong to $-1$ axis of Lambert
$W$ function, it has a lower bound corresponding to $\frac{1}{\lambda}$.
Thus, $T<\overline{S}^{NI}$ for any $\kappa\left(T,\lambda\right)$.
Second, consider $\lambda T\in\left[1,\tilde{\lambda T}\right]$, where $\tilde{\lambda T}$ solves $e^{\lambda T}=\lambda T\left(\lambda T+1\right)+1$. 
In this case, from (\ref{T_tilde_tilde}), if $\kappa\left(T,\lambda\right)\leq e^{-\lambda T}\lambda T$
($\kappa\left(T,\lambda\right)\geq e^{-\lambda T}\lambda T$, respectively), then
$T\leq\overline{S}^{NI}$ ($T\geq\overline{S}^{NI}$, respectively).
However, $\kappa^{NI}\left(T,\lambda\right)\leq e^{-\lambda T}\lambda T$.
Thus, $\kappa^{ND}\left(T,\lambda\right)=\kappa^{NI}\left(T,\lambda\right)$.

Case b. $e^{\lambda T}>\lambda T\left(\lambda T+1\right)+1.$ As before,
it holds that if $\kappa\left(T,\lambda\right)\leq e^{-\lambda T}\lambda T$
($\kappa\left(T,\lambda\right)\geq e^{-\lambda T}\lambda T$), then
$T\leq\overline{S}^{NI}$ ($T\geq\overline{S}^{NI}$, respectively).
Denote
\[
\kappa^{ND}\left(T,\lambda\right)\coloneqq e^{-\lambda T}\lambda T.
\]

As $\kappa^{NI}\left(T,\lambda\right)>\kappa^{ND}\left(T,\lambda\right)$,
two cases emerge. If $0<\kappa\leq\kappa^{ND}\left(T,\lambda\right)$,
then $T\leq\overline{S}^{NI}$, and from $\kappa\leq\kappa^{NI}\left(T,\lambda\right)$,
it holds that $S^{NI}=T$ and as the agent does not strictly benefit
from disclosing any information, she chooses non-disclosure. If $\kappa>\kappa^{ND}\left(T,\lambda\right)$,
then $T>\overline{S}^{NI}$ and the agent can potentially benefit
from information disclosure.

\emph{2. The case of $\kappa^{ND}\left(T,\lambda\right)<\kappa\leq\tilde{\kappa}\left(T,\lambda\right)$}.

The result is established in Proposition \ref{delayed_tau2}.

\emph{3. The case of $\tilde{\kappa}\left(T,\lambda\right)<\kappa<\kappa^{FI}\left(T,\lambda\right)$}.

The result is established in Proposition \ref{case3}.

\emph{4. The case of $\kappa\geq\kappa^{FI}\left(T,\lambda\right)$}. 

The principal's long-run payoff in the full-information benchmark non-positive. Thus, the agent can not strictly benefit from information disclosure and chooses non-disclosure.

\end{proof}

\begin{proof}[Proof of Lemma \ref{lemm_obedience}]
\emph{Necessity}. Assume $V_{t}(\tau)<0$ for some $t$. In that case, it is optimal for the principal to deviate to stopping at $t<\tau$. Thus, there is no information policy $\sigma$, for which this $\tau$ is the principal's best reply. Assume $V_{\tau}^{NI}\geq0$. Thus, the principal deviates to stopping at $t>\tau$, and there is no $\sigma$, for which this $\tau$ is the best reply.

\emph{Sufficiency}. Assume (\ref{obed_constr}) holds. $V_{t}\left(\tau\right)\geq0$ for all $t<\tau$ implies that the principal
prefers to continue rather than to stop the funding for all $t<\tau$. Thus, it can
not be that case that the principal stops before $\tau$. Further, $V_{\tau}^{NI}<0$ implies that, conditional on reaching the date of stopping $\tau$, it is better for the principal to stop immediately rather than to stop at $t>\tau$. 

Consider   a direct
recommendation mechanism $\sigma$ with $M=\left\{ 0,1\right\}$ such that whenever, based on the evolution of the state process, the considered investment schedule $\tau$ suggests stopping the funding, the direct recommendation mechanism sends the message $m=0$ to the principal. As it is optimal for the principal to stop at $\tau$, $\tau$ is the principal's best reply to $\sigma$. 
\end{proof}

\begin{proof}[Proof of Lemma \ref{kappa_tilde}] Consider the recommendation
mechanism immediately disclosing the completion of the second stage
of the project; it is given by $\tau=\tau_{2}\wedge T$. There exists
such $\tilde{\kappa}\left(T,\lambda\right)$ that solves the principal's
binding $t=0$ individual rationality constraint when $\tau=\tau_{2}\wedge T$:

\begin{equation}
V\left(\tau_{2}\right)=0,\label{kappa_cT}
\end{equation}
where 
\begin{equation}
\begin{aligned}V\left(\tau_{2}\right) & =p_{2}\left(T\right)v-\Esp\left[\tau_{2}\wedge T\right]c\\
 & =v\left(1-e^{-\lambda T}-\lambda Te^{-\lambda T}\right)-c\frac{1}{\lambda}\left(2-2e^{-\lambda T}-\lambda Te^{-\lambda T}\right).\label{V_tau2}
\end{aligned}
\end{equation}
The solution to equation (\ref{kappa_cT}) is given by 
\begin{equation}
\tilde{\kappa}\left(T,\lambda\right)=\frac{1-e^{\lambda T}+\lambda T}{2-2e^{\lambda T}+\lambda T}.
\end{equation}
Further, $\kappa>\tilde{\kappa}\left(T,\lambda\right)\Rightarrow V\left(\tau_{2}\right)<0$
and $\kappa\le\tilde{\kappa}\left(T,\lambda\right)\Rightarrow V\left(\tau_{2}\right)\geq0$.
\end{proof}

\begin{proof}[Proof of Lemma \ref{lemma_opt_sch1}]

\emph{Consider the case of $\kappa\in(\kappa^{ND}\left(T,\lambda\right),\kappa^{NI}\left(T,\lambda\right)]$}. The agent's relaxed problem for this case has the individual rationality constraints only for $t \in [0,\bar{S}^{NI}]$, and it  is given by 

\begin{equation}
\begin{aligned} & \max_{\tau\in\mathcal{T}}\left\{ c\cdot\Esp\left[\tau\right]\right\} \\
 & \text{s.t.}\,V_{t}\left(\tau\right)\geq0,\forall t\in\left[0,\bar{S}^{NI}\right],
\end{aligned}
\label{eq:relaxed_OP(lowest_k)}
\end{equation}
where $V_{t}(\tau)$ is given by (\ref{princ_cont_val}) and $\mathcal{T}$ is the set of stopping times with respect to the natural filtration of $x_{t}$.

Consider the candidate investment schedule $\tau$ such that $\tau\geq\bar{S}^{NI}\lor\left(\tau_{2}\land T\right)$
and $V\left(\tau\right)=V^{NI}$, where $V^{NI}$ is given by (\ref{NI_integral}).
I start with arguing that the candidate $\tau$ satisfies the system of individual rationality constraints. From Lemma \ref{no_feedb}, given candidate $\tau$,
the principal invests until $\bar{S}^{NI}$ with certainty and
the constraints in (\ref{eq:relaxed_OP(lowest_k)}) are satisfied for all $t\in[0,\bar{S}^{NI})$.
Further, $\tau$ implies that $V_{\bar{S}^{NI}}(\tau)=0$, i.e.,
the individual rationality constraint at $t=\bar{S}^{NI}$ is
binding. 

I proceed with arguing that the candidate $\tau$ maximizes the agent's objective function in (\ref{eq:relaxed_OP(lowest_k)}). The agent's objective can be WLOG written out as: 
\begin{equation}
\begin{array}{c}
W\left(\tau\right)=\underbrace{\p\left(x_{\tau}=2\right)v}_{\text{total surplus}}\,\,\,\,\,-\underbrace{V(\tau).}_{\text{principal's surplus}}\end{array}
\label{eq:total_surpl}
\end{equation}
By Lemma \ref{kappa_tilde}, an investment schedule $\tau$ that assigns
probability one to $\tau\geq\tau_{2}\land T$ satisfies the individual rationality constraint at $t=0$ in (\ref{eq:relaxed_OP(lowest_k)}). 
Note that, given $\tau\geq\tau_{2}\land T$, the total surplus in (\ref{eq:total_surpl}) is given
by $\p\left(x_{T}=2\right)v$, i.e., total surplus achieves its upper bound determined
by the exogenously given project deadline $T$. 
The principal's surplus in (\ref{eq:total_surpl}) is given by $V\left(\tau\right)=V^{NI}$, i.e., principal's surplus achieves its lower bound specified by (\ref{NI_integral}). This can be seen from the principal's decision problem, in which he best replies to an information policy $\sigma$. As $\sigma$ allows the principal to condition his actions on the information regarding the evolution of the state process, the principal's equilibrium payoff can not be lower than $V^{NI}$, his equilibrium payoff  when he is restricted to choosing actions without conditioning them on the information about the state process.  Thus, $\tau$ solves the relaxed problem (\ref{eq:relaxed_OP(lowest_k)}).

\emph{Consider the case of} $\kappa\in(\kappa^{NI}\left(T,\lambda\right),\tilde{\kappa}\left(T,\lambda\right)]$.
The agent's relaxed problem for this case has the individual rationality constraint only for the initial period, and it is given by 
\begin{equation}
\begin{aligned} & \max_{\tau\in\mathcal{T}}\left\{ c\cdot\Esp\left[\tau\right]\right\} \\
 & \text{s.t.}\,V(\tau) \geq 0,
\end{aligned}
\label{eq:relaxed_OP(low_k)}
\end{equation}
where $V(\tau)=\p\left(x_{\tau}=2\right)v-\Esp\left[\tau\right]c$.

Consider candidate investment schedule $\tau$ such that $\tau\geq\tau_{2}\land T$
and $V\left(\tau\right)=V^{NI}$.   For such $\tau$, agent's expected payoff (\ref{eq:total_surpl}) is given by $\p\left(x_{T}=2\right)v-V^{NI}$. As discussed for the parametric case $\kappa\in(\kappa^{ND}\left(T,\lambda\right),\kappa^{NI}\left(T,\lambda\right)]$, the first term is at its upper bound. To see that the second term is
at its lower bound, note that, from Lemma \ref{no_feedb}, $V^{NI}=0$, and thus the individual rationality constraint in (\ref{eq:relaxed_OP(low_k)})
is binding. Hence, $\tau$  solves the relaxed problem (\ref{eq:relaxed_OP(low_k)}).

\end{proof}

\begin{proof}[Proof of Proposition \ref{delayed_tau2}] The proof covers the case $\kappa\in(\kappa^{ND}\left(T,\lambda\right),\kappa^{NI}\left(T,\lambda\right)]$ and the case  $\kappa\in(\kappa^{NI}\left(T,\lambda\right),\tilde{\kappa}\left(T,\lambda\right)]$ separately.

\emph{1. The case of $\kappa\in(\kappa^{ND}\left(T,\lambda\right),\kappa^{NI}\left(T,\lambda\right)]$}.

I start with proving the existence of $S^{*}$ such that $V\left(\tau\right)=V^{NI}$. Assume that $S^{*}>\bar{S}^{NI}$. For all $t\in[\bar{S}^{NI},S^{*})$,
stopping never occurs, at $t=S^{*}$ it occurs if $x_{S^{*}}=2$,
and for all $t\in(S^{*},\tau)$ it occurs at $t=\tau_{2}\land T$. For
$t\in[S^{*},\tau)$, the absence of stopping induces posteriors $q_{n}\left(t\right)$.
Further, for $t\in[S^{*},\tau)$ the principal discounts future benefits
from postponing stopping using the probability of the state being
$2$. Hence, the continuation value at $t=\bar{S}^{NI}$ can be written as
\begin{equation}
V_{\bar{S}^{NI}}\left(\tau\right)=v\lambda\left(\int_{\bar{S}^{NI}}^{S^{*}}p_{1}\left(z\right)-\kappa dz+\int_{S^{*}}^{T}\left(q_{1}\left(z\right)-\kappa\right)\left(1-\p\left(x_{z}=2\right)\right)dz\right).\label{ex_S^*_V}
\end{equation}

The principal's long-run payoff is given by
\[
V\left(\tau\right)=\int_{0}^{\bar{S}^{NI}}\left(v\cdot p_{1}\left(s\right)\lambda-c\right)ds+V_{\bar{S}^{NI}}\left(\tau\right),
\]
where $\int_{0}^{\bar{S}^{NI}}\left(v\cdot p_{1}\left(s\right)\lambda-c\right)ds=V^{NI}$.  
Thus, to ensure that $S^{*}$ makes the individual rationality constraint bind at $t=\bar{S}^{NI}$, i.e., $V\left(\tau\right)=V^{NI}$, it is necessary and sufficient that  $V_{\bar{S}^{NI}}\left(\tau\right)=0$.
Using (\ref{ex_S^*_V}), this equation can be written as 
\[
\int_{\bar{S}^{NI}}^{S^{*}}\kappa-p_{1}\left(z\right)dz=\int_{S^{*}}^{T}\left(q_{1}\left(z\right)-\kappa\right)\left(1-\p\left(x_{z}=2\right)\right)dz.
\]
Let $g\left(S\right)\coloneqq\int_{\bar{S}^{NI}}^{S}\kappa-p_{1}\left(z\right)dz$
and $k\left(S\right)\coloneqq\int_{S}^{T}\left(q_{1}\left(z\right)-\kappa\right)\left(1-\p\left(x_{z}=2\right)\right)dz$,
$S\in[\bar{S}^{NI},\tau)$. $q_{1}\left(t\right)\geq\kappa$, for
all $t\in\left[S^{*},T\right)$. Thus, $g\left(\bar{S}^{NI}\right)=0,k\left(\bar{S}^{NI}\right)>0$.
Further, $p_{1}\left(t\right)<\kappa$, for all $t\in(\bar{S}^{NI},T]$.
Hence, $g\left(T\right)>0,k\left(T\right)=0$. Finally, $p_{1}\left(t\right)\leq\kappa$,
for all $t\in\left[\bar{S}^{NI},T\right]$ implies that $g^{\prime}\left(S\right)\geq0$,
for all $s\in\left[\bar{S}^{NI},T\right]$, and $q_{1}\left(t\right)\geq\kappa$,
for all $t\in\left[S^{*},T\right]$ implies that $k^{\prime}\left(S\right)\leq0$,
for all $s\in\left[S^{*},T\right]$. Thus, by the intermediate value
theorem, there exists $S^{*}$ solving $V_{\bar{S}^{NI}}\left(\tau\right)=0$. Thus, there exists $S^{*}>\bar{S}^{NI}$ such that principal's individual rationality constraint is binding at $t=\bar{S}^{NI}$.

I proceed with proving that the investment schedule $\tau$ satisfies the conditions in Lemma \ref{lemm_obedience} and thus it is obedient.  

\emph{First, consider} $t\leq \bar{S}^{NI}$. The principal's continuation value
for all $t\in[0,\bar{S}^{NI}]$ can be written as 
\[
V_{t}\left(\tau\right)=\int_{t}^{\bar{S}^{NI}}v\lambda\left(p_{1}\left(s\right)-\kappa\right)ds+V_{\bar{S}^{NI}}\left(\tau\right).
\]
Given the binding individual rationality constraint, it becomes 
\[
V_{t}\left(\tau\right)=\int_{t}^{\bar{S}^{NI}}v\lambda\left(p_{1}\left(s\right)-\kappa\right)ds,\text{ for all }t\in[0,\bar{S}^{NI}).
\]

Finally, note that $V_{t}\left(\tau\right)$ above is equivalent to
$V_{t}^{NI}$ given by (\ref{V_NI_integr}). Lemma \ref{no_feedb}
implies that given $\kappa\in(\kappa^{ND}\left(T,\lambda\right),\kappa^{NI}\left(T,\lambda\right)]$,
$V^{NI}\left(0\right)=V\left(\tau\right)\geq0$. Further, Lemma \ref{NI_obed}
implies that $V\left(\tau\right)\geq0\Rightarrow V_{t}\left(\tau\right)\geq0,\forall t\in[0,\bar{S}^{NI})$.

\emph{Second, consider} $t\in\left[\bar{S}^{NI},S^{*}\right]$. Given $\kappa\in(\kappa^{ND}\left(T,\lambda\right),\kappa^{NI}\left(T,\lambda\right)]$, $p_{1}\left(t\right)\leq\kappa,\forall t\in\left[\bar{S}^{NI},S^{*}\right]$. Thus, $V_{t}^{NI}=0,\forall t\in\left[\bar{S}^{NI},S^{*}\right]$. 
The principal's continuation value is given by
\begin{equation}
V_{t}\left(\tau\right)=\int_{t}^{S^{*}}v\lambda\left(p_{1}\left(s\right)-\kappa\right)ds+V_{S^{*}}\left(\tau\right).\label{ic_ex_mech_S^*}
\end{equation}
As $p_{1}\left(t\right)\leq\kappa,\forall t\in\left[\bar{S}^{NI},S^{*}\right]$, $\int_{t}^{S^{*}}v\lambda\left(p_{1}\left(s\right)-\kappa\right)ds\leq0$  and it is increasing in $t$. As $V_{\bar{S}^{NI}}\left(\tau\right)=0$, where $V_{\bar{S}^{NI}}\left(\tau\right)$  is given by (\ref{ex_S^*_V}), it follows that $V_{t}\left(\tau\right)\geq0,\forall t\in\left[\bar{S}^{NI},S^{*}\right]$.

\emph{Third, consider} $t\in[S^{*},\tau)$. The absence of stopping at $t\geq S^{*}$
reveals that $x_{t}\neq2$. Thus, $q_{1}\left(t\right)=\frac{p_{1}\left(t\right)}{p_{0}\left(t\right)+p_{1}\left(t\right)}=\frac{\lambda t}{1+\lambda t}$,
$\forall t\in[S^{*},\tau)$, and, thus, $\dot{q}_{1}\left(t\right)>0$.
Further, $q_{1}\left(S^{*}\right)>\kappa$. The continuation value
$\forall t\in[S^{*},\tau)$ is given by 
\[
\begin{array}{c}
V_{t}\left(\tau\right)=\Esp\left[\int_{t}^{\tau}v\lambda\left(q_{1}\left(z\right)-\kappa\right)dz\,|\,t<\tau\right].\end{array}
\]
Thus, $V_{t}\left(\tau\right)\geq0$, $\forall t\in[S^{*},\tau)$.

\emph{2. The case of $\kappa^{NI}\left(T,\lambda\right)<\kappa\leq\tilde{\kappa}\left(T,\lambda\right)$}.

I start with proving the existence of $S^{*}$ such that $V\left(\tau\right)=0$. For all $t\in[0,S^{*})$, stopping never occurs,
at $t=S^{*}$ it occurs if $x_{S^{*}}=2$, and for all $t\in(S^{*},T]$
it occurs at $t=\tau_{2}\land T$. The principal's long-run payoff
can be written as 
\begin{equation}
V\left(\tau\right)=v\lambda\left(\int_{0}^{S^{*}}p_{1}\left(z\right)-\kappa dz+\int_{S^{*}}^{T}\left(q_{1}\left(z\right)-\kappa\right)\left(1-\p\left(x_{z}=2\right)\right)dz\right).\label{ex_S^*_V_1}
\end{equation}

To ensure that $S^{*}$ makes the individual rationality constraint
bind at $t=0$, it is necessary and sufficient that $V\left(\tau\right)=0$.
The next step of the proof consist of
inspecting (\ref{ex_S^*_V_1}) to establish that there exists $S^{*}$ ensuring that $V\left(\tau\right)=0$. It follows the respective part from the proof for the parametric case $\kappa^{ND}\left(T,\lambda\right)<\kappa\leq\kappa^{NI}\left(T,\lambda\right)$, imposing $\bar{S}^{NI}=0$
in it everywhere; thus, I omit it for the sake of space.

I proceed with proving that the investment schedule $\tau$ satisfies the conditions in Lemma \ref{lemm_obedience} and thus it is obedient. 
The principal's continuation value is given by
(\ref{ic_ex_mech_S^*}). As $\kappa\in(\kappa^{NI}\left(T,\lambda\right),\tilde{\kappa}\left(T,\lambda\right)]$, it follows from Lemma \ref{no_feedb} that $V_{t}^{NI}=0,\forall t\in\left[0,S^{*}\right]$. 
\emph{First, assume} $S^{*}\leq\bar{S}^{NI}$. From the proof of Lemma \ref{no_feedb}, it follows that $p_{1}\left(t\right)\leq\kappa,\forall t\in[0,\overline{S}]$, and $p_{1}\left(t\right)\geq\kappa,\forall t\in[\overline{S},\bar{S}^{NI}]$.
Thus, 
\begin{equation}
\int_{t}^{\bar{S}^{NI}}v\lambda\left(p_{1}\left(s\right)-\kappa\right)ds\geq\int_{0}^{\bar{S}^{NI}}v\lambda\left(p_{1}\left(s\right)-\kappa\right)ds,\forall t[0,\bar{S}^{NI}].\label{eqn:int_compar} 
\end{equation}
As $V_{t}(\tau)$ is given by (\ref{ic_ex_mech_S^*}), $V(\tau)=0$ and (\ref{eqn:int_compar}) imply that $V_{t}(\tau)\geq0,\forall t\in [0,S^{*}]$.  
\emph{Second, assume} $S^{*}\geq\bar{S}^{NI}$.
As $V(\tau)=0$ and $\int_{0}^{\bar{S}^{NI}}v\lambda\left(p_{1}\left(s\right)-\kappa\right)ds<0$, it must be that $V(\bar{S}^{NI})>0$. Further, $\int_{t}^{S^{*}}v\lambda\left(p_{1}\left(s\right)-\kappa\right)ds$ increases in $t$ for $t\in[\bar{S}^{NI},S^{*}]$. Thus, $V_{t}(\tau)\geq0,\forall t\in[0,S^{*}]$. 

Finally, the proof that  $V_{t}(\tau)\geq 0, \forall t\in[S^{*},\tau)$ follows the  the respective part of the proof for the parametric case $\kappa\in(\kappa^{ND}\left(T,\lambda\right),\kappa^{NI}\left(T,\lambda\right)]$; thus, I omit it for the sake of space.

\end{proof}

\begin{proof}[Proof of Lemma \ref{delay}]
I provide the proof for the parametric cases $\kappa^{ND}\left(T,\lambda\right)<\kappa\leq\kappa^{NI}\left(T,\lambda\right)$ and $\kappa^{NI}\left(T,\lambda\right)<\kappa\leq\tilde{\kappa}\left(T,\lambda\right)$ separately.

\emph{1. The case of $\kappa^{ND}\left(T,\lambda\right)<\kappa\leq\kappa^{NI}\left(T,\lambda\right)$}.

Under any obedient optimal policy, the principal's individual rationality constraint is binding, thus,
$V\left(\tau\right)=V^{NI}$, or equivalently $p_{2}\left(T\right)v-\Esp\left[\tau\right]c=p_{2}\left(\bar{S}^{NI}\right)v-\bar{S}^{NI}c.$
Thus,
\[
\Esp\left[\tau\right]=\frac{1}{\lambda\kappa}\left(p_{2}\left(T\right)-p_{2}\left(\bar{S}^{NI}\right)\right)+\bar{S}^{NI}.
\]
Differentiating both sides with respect to $\kappa$ yields
\[
\frac{\partial\Esp\left[\tau\right]}{\partial\kappa}=\frac{e^{-T\lambda}\left(1+T\lambda\right)-e^{-\bar{S}^{NI}\lambda}-\kappa}{\kappa^{2}\lambda}.
\]
The equation
\[
e^{-T\lambda}\left(1+T\lambda\right)-e^{-\bar{S}^{NI}\lambda}-\kappa=0
\]
can be equivalently rewritten as
\[
e^{-T\lambda}-e^{-\bar{S}^{NI}\lambda}=\kappa-e^{-T\lambda}T\lambda.
\]
It has a unique solution corresponding to $\kappa=\kappa^{ND}\left(T,\lambda\right)\coloneqq e^{-T\lambda}T\lambda.$
As $\kappa>\kappa^{ND}\left(T,\lambda\right),$ it holds that
$\partial\Esp\left[\tau\right]/\partial\kappa<0$.

\emph{2. The case of $\kappa^{NI}\left(T,\lambda\right)<\kappa\leq\tilde{\kappa}\left(T,\lambda\right)$}.

The principal's long-run payoff under any obedient optimal policy is given by 
\[
\Esp\left[\tau\right]c=p_{2}\left(T\right)v.
\]
Rewriting it equivalently, $\Esp\left[\tau\right]=\frac{1}{\lambda}\frac{1}{\kappa}p_{2}\left(T\right)\Rightarrow$
$\partial\Esp\left[\tau\right]/\partial\kappa<0$.

\end{proof}

\begin{proof}[Proof of Lemma \ref{lemma_effic_sch2}]

Lemma \ref{kappa_tilde} implies that if a schedule $\tau$ assigns
zero probability to stopping in states $0$ and $1$ then $V\left(\tau\right)<0$
and the individual rationality constraint is violated. Thus, the necessary
condition for a schedule $\tau$ to be individually rational under $\kappa\in(\tilde{\kappa}\left(T,\lambda\right),\kappa^{FI}\left(T,\lambda\right))$
is that it assigns a positive probability to stopping not only in
state $2$, but also to stopping in either state $0$ or state $1$.
Consider a schedule $\tau$ that assigns a positive probability to
stopping in state $1$. Consider an alternative schedule $\tau^{\prime}$
which is induced by reallocating the probability mass of stopping
in state $1$ to stopping at $\tau_{2}\land T$. Lemma \ref{lemma_FI}
suggests that in state $1$ the principal strictly benefits from postponing
the stopping until the second stage of the project is completed. Thus,
$V\left(\tau^{\prime}\right)>V\left(\tau\right)$. Further, under
$\tau^{\prime}$ the principal invests strictly longer, in expectation.
Thus, $W\left(\tau^{\prime}\right)>W\left(\tau\right)$. Thus, for
a schedule to be optimal it should not assign a positive probability
to stopping in state $1$.

Next, consider a schedule $\tau$ which assigns a positive probability
to stopping in states $0$ and $2$. Assume that the stopping in state
$0$ happens at date $S$, which can be either deterministic or stochastic:
if $x_{S}=0$ then $\tau=S$, otherwise, $\tau\geq\tau_{2}\land T$
and there exists $\omega\in\Omega$ such that $\tau\left(\omega\right)>\tau_{2}\left(\omega\right)$,
i.e., with a positive probability, stopping in state $2$ happens strictly
after the date of transition to state $2$. Assume that $V\left(\tau\right)=0$.
Consider the following investment schedule $\tilde{\tau}$: if $x_{\tilde{S}}=0$
then $\tilde{\tau}=\tilde{S}$, $\Esp{[\tilde{S}]}>\Esp{[S]}$, otherwise, $\tilde{\tau}=\tau_{2}\land T$,
and $V\left(\tilde{\tau}\right)=0$. Further, from (\ref{zero_sum_pyf}),
the agent's objective is given by 
\[
\begin{aligned}W\left(\tilde{\tau}\right)-W\left(\tau\right) & =\left(SV\left(\tilde{\tau}\right)-V\left(\tilde{\tau}\right)\right)-\left(SV\left(\tau\right)-V\left(\tau\right)\right)\\
 & =SV\left(\tilde{\tau}\right)-SV\left(\tau\right).
\end{aligned}
\]
The change from $\tau\geq\tau_{2}\land T$ to $\tau=\tau_{2}\land T$
induces no loss in total surplus as the measure of $\omega\in\Omega$
satisfying the event $\left\{ \tau_{2}\leq T\right\} $ is equal for
both schedules. Further, the change from conditional stopping at $S$
to conditional stopping at $\tilde{S}$
induces an increase in total surplus as $\p\left(x_{\tilde{S}}=0\right)<\p\left(x_{S}=0\right)$
and thus, in the latter case, conditional stopping happens less frequently.
Hence, $SV\left(\tilde{\tau}\right)\geq SV\left(\tau\right)$. Thus,
for a schedule that assigns positive probability to stopping in states
$0$ and $2$ to be optimal, it is necessary that stopping in state
$2$ happens at $\tau_{2}$ with probability $1$.

\end{proof}

\begin{proof}[Proof of Proposition \ref{case3}]

Given Lemma \ref{lemma_effic_sch2}, the space of candidate optimal investment schedules under $\kappa\in(\tilde{\kappa}\left(T,\lambda\right),\kappa^{FI}\left(T,\lambda\right)]$
simplifies to schedules such that stopping in state $2$ happens at
$\tau_{2}$, and also stopping in state $0$ happens with positive
probability. Thus, to characterize the optimal schedule under $\kappa\in(\tilde{\kappa}\left(T,\lambda\right),\kappa^{FI}\left(T,\lambda\right)]$,
I need to characterize the assignment of the probability mass of stopping
in state $0$ that is optimal for the agent given the principal's individual rationality constraints. To do this, I consider the agent's optimal design of a \emph{device that randomizes over the dates of stopping in state $0$}. 

At $t=0$, the agent chooses a distribution
$F_{\rho}$ on $\left[0,T\right]$, observable to both the agent and the principal. $\rho$  stands for the random date at which the
stopping occurs if the state is $0$ by that date. $\rho$  is drawn at $t=0$ according to  $F_{\rho}$, which is independent from the state
process $x_{t}$,  and the draw privately observed by the agent. 


To formulate the agent's design problem, I start with characterizing the welfare implications of stopping in
state $0$ for the agent and principal. A few useful objects
are $SV_{t|0}\left(\tau_{2}\right)$ and $V_{t|0}\left(\tau_{2}\right)$.
$SV_{t|0}\left(\tau_{2}\right)$ is the time $t$ continuation total surplus
given that $x_{t}=0$ at $t$ and completion of the second stage of
the project is immediately disclosed whenever it happens, $\tau=\tau_{2}\land T$:

\begin{equation}
\begin{aligned}SV_{t|0}\left(\tau_{2}\right) & =v\p\left(\tau_{2}\leq T|x_{t}=0\right)\\
 & =v\left[1-e^{-\lambda\left(T-t\right)}-\lambda\left(T-t\right)e^{-\lambda\left(T-t\right)}\right].
\end{aligned}
\label{L_function}
\end{equation}

$V_{t|0}\left(\tau_{2}\right)$ is the principal's time $t$ continuation
value given that $x_{t}=0$ and completion of the second stage of
the project is immediately disclosed, $\tau=\tau_{2}\land T$: 
\[
V_{t|0}\left(\tau_{2}\right)=v\p\left(\tau_{2}\leq T|x_{t}=0\right)-c\Esp\left[\tau_{2}\wedge T-t|x_{t}=0\right],
\]
where $v\p\left(\tau_{2}\leq T|x_{t}=0\right)$ is given by (\ref{L_function})
and

\begin{equation}
\begin{aligned} & \Esp\left[\tau_{2}\land T-t|x_{t}=0\right]\\
= & \p\left(\tau_{2}\leq T|x_{t}=0\right)\int_{t}^{T}z\cdot\frac{\lambda^{2}\left(z-t\right)e^{-\lambda\left(z-t\right)}}{\p\left(\tau_{2}\leq T|x_{t}=0\right)}dz+\p\left(\tau_{2}>T|x_{t}=0\right)T-t\\
= & \frac{2}{\lambda}-\frac{2}{\lambda}e^{-\lambda\left(T-t\right)}-e^{-\lambda\left(T-t\right)}\left(T-t\right).
\end{aligned}
\label{S_function}
\end{equation}

I proceed with a useful lemma.

\begin{lemma}\label{decomposition}

Given an investment schedule 
\begin{equation}
\tau=\begin{cases}
\rho, & \text{if }x_{\rho}=0\\
\tau_{2}\land T, & \text{otherwise},
\end{cases}\label{mech_with_S}
\end{equation}
where $\rho$ has a publicly observable distribution $F_{\rho}$ on $\left[0,T\right]$, $\rho$
is independent of the state process $x_{t}$ and is drawn at $t=0$, and the draw is unobservable to the players, the
total surplus at date $t$ can be written as 
\[
SV_{t}\left(\tau\right)=SV_{t}\left(\tau_{2}\right)-\Esp_{F_{\rho}}\left[\p\left(x_{\rho}=0|t<\tau\right)SV_{\rho|0}\left(\tau_{2}\right)\right],
\]
and the principal's expected payoff at date $t$ can be written as 
\[
V_{t}\left(\tau\right)=V_{t}\left(\tau_{2}\right)-\Esp_{F_{\rho}}\left[\p\left(x_{\rho}=0|t<\tau\right)V_{\rho|0}\left(\tau_{2}\right)\right],
\]
for all $t\geq0$.

\end{lemma}

\begin{proof}

By construction, $SV_{t}\left(\tau\right)$ corresponds to the expected
value of the project completion payoff under stopping policy $\tau$
conditional on stopping not having happened by $t$, i.e., $t<\tau$.
Given (\ref{mech_with_S}), the principal gets $v$ either if the
second stage is completed before $\rho$ or if the first stage is
completed before $\rho$ and the second stage is completed before
$T$. Note that when $t<\rho$, $t<\tau$ implies that the state is
either $0$ or $1$, and, when $t\geq\rho$, $t<\tau$ implies that
the state is $1$. Thus,

\[
SV_{t}\left(\tau\right)=v\Esp_{F_{\rho}}\bigg[\p\big(\left\{ x_{\rho}=1\right\} \cap\left\{ \tau_{2}\leq T\right\} |t<\tau\big)+\p\big(x_{\rho}=2|t<\tau\big)\bigg].
\]
Further, for each realization of $\rho$,

\[
\p\big(\left\{ x_{\rho}=1\right\} \cap\left\{ \tau_{2}\leq T\right\} |t<\tau\big)=\p\left(x_{\rho}=1|t<\tau\right)\p\left(\tau_{2}\leq T|x_{\rho}=1\right).
\]
Thus,

\begin{equation}
SV_{t}\left(\tau\right)=v\Esp_{F_{\rho}}\bigg[\p\big(x_{\rho}=1|t<\tau\big)\p\big(\tau_{2}\leq T|x_{\rho}=1\big)+\p\big(x_{\rho}=2|t<\tau\big)\bigg].\label{SV_0_iii}
\end{equation}

$SV_{\rho|0}\left(\tau_{2}\right)$ corresponds to the expected value
of the project completion payoff when $x_{\rho}=0$. In that case,
$v$ is obtained when the completion of the second stage happens not
later than $T$. Thus, $SV_{\rho|0}\left(\tau_{2}\right)=\Esp_{F_{\rho}}\left[v\p\left(\tau_{2}\leq T|x_{\rho}=0\right)\right]$.
Therefore,

\begin{equation}
\begin{aligned} & SV_{t}\left(\tau_{2}\right)-\Esp_{F_{\rho}}\big[\p\left(x_{\rho}=0|t<\tau\right)SV_{\rho|0}\left(\tau_{2}\right)\big]\\
= & \p\left(x_{T}=2|t<\tau\right)v-\Esp_{F_{\rho}}\big[\p\left(x_{\rho}=0|t<\tau\right)v\p\left(\tau_{2}\leq T|x_{\rho}=0\right)\big]\\
= & v\Esp_{F_{\rho}}\big[\p\left(x_{T}=2|t<\tau\right)-\p\left(x_{\rho}=0|t<\tau\right)\p\left(\tau_{2}\leq T|x_{\rho}=0\right)\big].
\end{aligned}
\label{SV_0_iii_proof}
\end{equation}
Thus, given (\ref{SV_0_iii}) and (\ref{SV_0_iii_proof}), to complete
the proof of the first result of the Lemma \ref{decomposition}, it suffices to show that,

\[
\begin{aligned} & \p\left(x_{T}=2|t<\tau\right)-\p\left(x_{\rho}=0|t<\tau\right)\p\left(\tau_{2}\leq T|x_{\rho}=0\right)\\
= & \p\left(x_{\rho}=2|t<\tau\right)+\p\left(x_{\rho}=1|t<\tau\right)\p\left(\tau_{2}\leq T|x_{\rho}=1\right)
\end{aligned}
\]
Using the full probability formula, 
\[
\begin{aligned}\p\left(x_{T}=2|t<\tau\right)=\\
 & \p\left(x_{\rho}=0|t<\tau\right)\p\left(\tau_{2}\leq T|x_{\rho}=0\right)\\
 & +\p\left(x_{\rho}=1|t<\tau\right)\p\left(\tau_{2}\leq T|x_{\rho}=1\right)\\
 & +\p\left(x_{\rho}=2|t<\tau\right)\p\left(\tau_{2}\leq T|x_{\rho}=2\right).
\end{aligned}
\]
Hence, 
\begin{equation}
SV_{t}\left(\tau\right)=SV_{t}\left(\tau_{2}\right)-\Esp_{F_{\rho}}\big[\p\big(x_{\rho}=0|t<\tau\big)SV_{\rho|0}\left(\tau_{2}\right)\big],\text{ for all }t\geq0.\label{SV_decomp}
\end{equation}

I proceed with proving the second result of Lemma \ref{decomposition}.
First, applying (\ref{SV_decomp}) to $V_{t}\left(\tau\right)$ yields

\begin{equation}
\begin{aligned}V_{t}\left(\tau\right) & =SV_{t}\left(\tau\right)-\Esp_{F_{\rho}}\big[c\Esp\left[\tau|t<\tau\right]\big]\\
 & =SV_{t}\left(\tau_{2}\right)-\Esp_{F_{\rho}}\big[\p\big(x_{\rho}=0|t<\tau\big)SV_{\rho|0}\left(\tau_{2}\right)-c\Esp\left[\tau|t<\tau\right]\big].
\end{aligned}
\label{V0_decomp}
\end{equation}
Further, for each realization of $\rho$: 
\begin{equation}
\begin{aligned} & \Esp\left[\tau|t<\tau\right]\\
= & \p\left(x_{\rho}=0|t<\tau\right)\Esp\left[\tau|x_{\rho}=0\right]\\
 & +\p\big(x_{\rho}=1|t<\tau\big)\Esp\left[\tau|x_{\rho}=1\right]+\p\big(x_{\rho}=2|t<\tau\big)\Esp\left[\tau|x_{\rho}=2\right]\\
= & \p\left(x_{\rho}=0|t<\tau\right)\rho\\
 & +\p\big(x_{\rho}=1|t<\tau\big)\Esp\left[\tau_{2}\land T|x_{\rho}=1\right]+\p\big(x_{\rho}=2|t<\tau\big)\Esp\left[\tau_{2}\land T|x_{\rho}=2\right]\\
= & \p\big(x_{\rho}=0|t<\tau\big)\rho+\Esp\left[\tau_{2}\land T|t<\tau\right]-\p\big(x_{\rho}=0|t<\tau\big)\Esp\left[\tau_{2}\land T|x_{\rho}=0\right]\\
= & \Esp\left[\tau_{2}\land T|t<\tau\right]-\p\big(x_{\rho}=0|t<\tau\big)\big(\Esp\left[\tau_{2}\land T|x_{\rho}=0\right]-\rho\big),
\end{aligned}
\label{Etau_decomp}
\end{equation}
where the second equality uses the full probability formula.

Plugging (\ref{Etau_decomp}) into (\ref{V0_decomp}) yields 
\[
\begin{aligned} & SV_{t}\left(\tau_{2}\right)-\Esp_{F_{\rho}}\left[c\Esp\left[\tau_{2}\land T|t<\tau\right]\right]\\
 & -\Esp_{F_{\rho}}\left[\p\left(x_{\rho}=0|t<\tau\right)\left(SV_{\rho|0}\left(\tau_{2}\right)-c\Esp\left[\tau_{2}\land T-\rho|x_{\rho}=0\right]\right)\right]\\
= & V_{t}\left(\tau_{2}\right)-\Esp_{F_{\rho}}\left[\p\left(x_{\rho}=0|t<\tau\right)V_{\rho|0}\left(\tau_{2}\right)\right],\forall t\geq0.
\end{aligned}
\]

\end{proof}

I proceed to formulating the agent's problem. The agent's objective can be represented
as

\[
\begin{array}{c}
c\Esp\left[\tau\right]=SV\left(\tau\right)-V\left(\tau\right).\end{array}
\]
Using Lemma \ref{decomposition}, 
\begin{equation}
\begin{aligned} & SV\left(\tau\right)-V\left(\tau\right)\\
= & SV\left(\tau_{2}\right)-V\left(\tau_{2}\right)-\Esp_{F_{\rho}}\left[\p\left(x_{\rho}=0\right)\left(V_{\rho|0}\left(\tau_{2}\right)-SV_{\rho|0}\left(\tau_{2}\right)\right)\right]\\
= & SV\left(\tau_{2}\right)-V\left(\tau_{2}\right)-c\Esp_{F_{\rho}}\left[\p\left(x_{\rho}=0\right)\Esp\left[\tau_{2}\land T-\rho|x_{\rho}=0\right]\right].
\end{aligned}
\label{Case_c_obj_w_rho}
\end{equation}
The individual rationality constraint for the principal can be expressed
as 
\begin{equation}
V_{t}\left(\tau\right)\geq0,\forall t\geq0 \iff V_{t}\left(\tau_{2}\right)\geq\Esp_{F_{\rho}}\left[\p\left(x_{\rho}=0|t<\tau\right)V_{\rho|0}\left(\tau_{2}\right)\right],\forall t\geq0. \label{CaseC_constr}
\end{equation}
Finally, (\ref{Case_c_obj_w_rho}) yields the objective and (\ref{CaseC_constr})
yields the individual rationality constraint for the agent's problem 
\begin{equation}
\begin{aligned} & \min_{F_{\rho}}\left\{\Esp_{F_{\rho}}\left[\p\left(x_{\rho}=0\right)\Esp\left[\tau_{2}\land T-\rho|x_{\rho}=0\right]\right]\right\}\\
 & \text{s.t.}\,\Esp_{F_{\rho}}\left[\p\left(x_{\rho}=0|t<\tau\right)\left(c\Esp\left[\tau_{2}\land T-\rho|x_{\rho}=0\right]-SV_{\rho|0}\left(\tau_{2}\right)\right)\right]\geq-V_{t}\left(\tau_{2}\right),\forall t\geq0.
\end{aligned}
\label{CaseC_fullOP}
\end{equation}

I proceed in two steps: first, I formulate and \emph{solve the relaxed version of} (\ref{CaseC_fullOP}); second, I demonstrate that the solution to the relaxed problem \emph{satisfies the full system of constraints in} (\ref{CaseC_fullOP}). The relaxed problem has the principal's individual rationality constraint only for $t=0$:

\begin{equation}
\begin{aligned} & \min_{F_{\rho}}\left\{\Esp_{F_{\rho}}\left[\p\left(x_{\rho}=0\right)\Esp\left[\tau_{2}\land T-\rho|x_{\rho}=0\right]\right]\right\}\\
 & \text{s.t.}\,\Esp_{F_{\rho}}\left[\p\left(x_{\rho}=0\right)\left(c\Esp\left[\tau_{2}\land T-\rho|x_{\rho}=0\right]-SV_{\rho|0}\left(\tau_{2}\right)\right)\right]\geq-V\left(\tau_{2}\right).
\end{aligned}
\label{CaseC_relaxedOP}
\end{equation}

The Lagrangian function for the problem is 
\[
\begin{aligned}\mathcal{L}= & \Esp_{F_{\rho}}\left[\p\left(x_{\rho}=0\right)\Esp\left[\tau_{2}\land T-\rho|x_{\rho}=0\right]\right]\\
 & -\mu\left(\Esp_{F_{\rho}}\left[\p\left(x_{\rho}=0\right)\left(c\Esp\left[\tau_{2}\land T-\rho|x_{\rho}=0\right]-SV_{\rho|0}\left(\tau_{2}\right)\right)\right]+V\left(\tau_{2}\right)\right),
\end{aligned}
\]
where $\p\left(x_{\rho}=0\right)=e^{-\lambda \rho}$, $\Esp\left[\tau_{2}\land T-\rho|x_{\rho}=0\right]$
is given by (\ref{S_function}), $SV_{\rho|0}\left(\tau\right)$ is given
by (\ref{L_function}).

I obtain the F.O.C., which needs to hold for each value of $\rho$
that has a positive probability in $F_{\rho}$: 
\begin{equation}
e^{-\lambda T}\left(c\left(2e^{-\lambda\left(T-\rho\right)}-1\right)\left(\mu-1\right)-\mu\lambda v\left(e^{-\lambda\left(T-\rho\right)}-1\right)\right)=0.\label{FOC}
\end{equation}
The derivative of the left-hand side of (\ref{FOC}) w.r.t. $\rho$ is
given by $e^{-\lambda \rho}\lambda\left(2c+\mu\left(\lambda v-2c\right)\right)$.
As $\kappa^{FI}\left(T,\lambda\right)<\frac{1}{2}$, the derivative
is positive. Thus, there exists at most one $\rho$ that satisfies the
FOC (\ref{FOC}). Thus, the optimal $F_{\rho}$ is degenerate. I denote it with $S_{0}^{A}$, the interim deadline.

I proceed with characterizing the optimal $S_{0}^{A}$: 
\begin{equation}
\begin{aligned} & \min_{S\in\left[0,T\right]}\left\{\p\left(x_{S}=0\right)\Esp\left[\tau_{2}\land T-S|x_{S}=0\right]\right\}\\
 & \text{s.t.}\,\p\left(x_{S}=0\right)\left(c\Esp\left[\tau_{2}\land T-S|x_{S}=0\right]-SV_{S|0}\left(\tau_{2}\right)\right)\geq-V\left(\tau_{2}\right).
\end{aligned}
\label{OP_S_INT_red}
\end{equation}
The system of F.O.C. is given by

\[
\begin{cases}
\begin{aligned} & e^{-\lambda T}c\left(2e^{-\lambda\left(T-S\right)}-1\right)\left(\mu-1\right)\\
 & -e^{-\lambda T}\mu\lambda v\left(e^{-\lambda\left(T-S\right)}-1\right)
\end{aligned}
 & \begin{aligned} & \geq0\,\text{ if }S=0\\
 & =0\,\text{ if }S\in\left(0,T\right)\\
 & \leq0\,\text{ if }S=T
\end{aligned}
\\
\\
\begin{aligned} & \frac{c}{\lambda}e^{-\lambda T}\left(2\left(e^{-\lambda\left(T-S\right)}-1\right)-\lambda\left(T-S\right)\right)\\
 & -ve^{-\lambda T}\left(\left(e^{-\lambda\left(T-S\right)}-1\right)-\lambda\left(T-S\right)\right)+V\left(\tau_{2}\right)\geq0
\end{aligned}
 & =0\,\text{ if }\mu>0.
\end{cases}
\]
Assume $\mu=0$. In this case, the first F.O.C. wrt $S$ yields
$-ce^{-\lambda T}\left(2e^{-\lambda\left(T-S\right)}-1\right)$.
The expression is negative for all $S\in\left(0,T\right)$. Thus,
$\mu>0$, and optimal $S$ solves the binding constraint. Thus,
I proceed with inspecting the corresponding equation given by 
\begin{equation}
\begin{aligned} & \frac{c}{\lambda}e^{-\lambda T}\left(2\left(e^{-\lambda\left(T-S\right)}-1\right)-\lambda\left(T-S\right)\right)\\
 & -ve^{-\lambda T}\left(\left(e^{-\lambda\left(T-S\right)}-1\right)-\lambda\left(T-S\right)\right)\\
= & -V\left(\tau_{2}\right),
\end{aligned}
\label{binding_constr}
\end{equation}
where $V\left(\tau_{2}\right)$ is given by (\ref{V_tau2}).

The solution to (\ref{binding_constr}) is given by 
\begin{equation}
S=\frac{1}{\lambda}\left[\gamma+\mathcal{W}\left(-\gamma e^{-\gamma}\right)\right],\label{S_INT_gamma}
\end{equation}
where $\gamma=e^{\lambda T}\frac{1-2\kappa}{1-\kappa}$
and $\mathcal{W}(.)$ denotes the Lambert $W$ function.

Denote the $0$ and $-1$ branches of the Lambert $W$ function by
$\mathcal{W}_{0}(.)$ and $\mathcal{W}_{-1}(.)$. $\kappa\in\left(0,\frac{1}{2}\right)$,
thus, $\gamma>0$. (\ref{S_INT_gamma}) depends on $\gamma$ and for
each $\gamma\neq1$ corresponds to two points as the Lambert $W$
function has two branches. The values of (\ref{S_INT_gamma}) as a
function of $\gamma$ are presented in Figure \ref{Lambert}. They
are given by 
\[
S=\begin{cases}
\left(\frac{1}{\lambda}\left[\gamma+\mathcal{W}_{-1}\left(-\gamma e^{-\gamma}\right)\right],0\right), & \text{if }\gamma<1\\
\left(0,\frac{1}{\lambda}\left[\gamma+\mathcal{W}_{0}\left(-\gamma e^{-\gamma}\right)\right]\right), & \text{if }\gamma>1\\
0, & \text{if }\gamma=1.
\end{cases}
\]

\begin{figure}[H]
\captionsetup{justification=centering} \centering \includegraphics[width=0.6\textwidth]{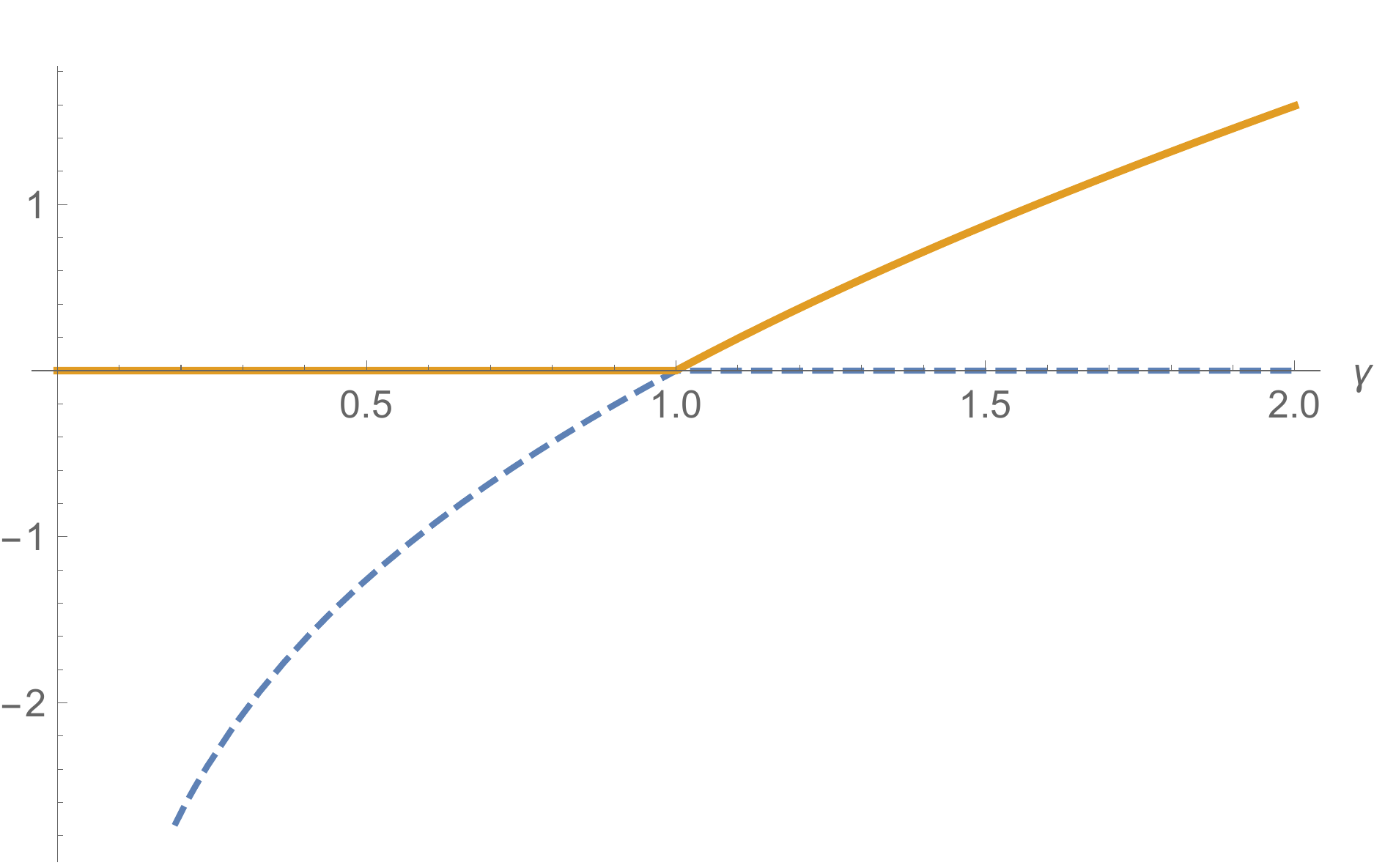}
\caption{Roots of equation (\ref{binding_constr}) as a function of the parameter
$\gamma$: \protect \protect \protect \\
 root corresponding to branch $0$ of the Lambert $W$ function -
\textbf{thick}; \protect \protect \protect \\
 root corresponding to branch $-1$ of the Lambert $W$ function -
\textbf{dashed}.}
\label{Lambert} 
\end{figure}

$\gamma$ is decreasing in $\kappa$, and $\gamma_{|\kappa=\kappa^{FI}}=1$.
As $\kappa\leq\kappa^{FI}$, which corresponds to $\gamma\geq1$,
the solution to (\ref{binding_constr}) is given by 
\[
\begin{array}{cc}
S_{A}=0,\,\,\,\,\,\,\,\,\, & S_{B}=\frac{1}{\lambda}\left[\gamma+\mathcal{W}_{0}\left(-\gamma e^{-\gamma}\right)\right].\end{array}
\]
As the objective of (\ref{OP_S_INT_red}) is decreasing in $S$ and
$S_{B}>S_{A}$, the solution to (\ref{OP_S_INT_red}) is given by
\begin{equation}
S_{0}^{A}=\frac{1}{\lambda}\left[\gamma+\mathcal{W}_{0}\left(-\gamma e^{-\gamma}\right)\right],\gamma=e^{\lambda T}\frac{1-2\kappa}{1-\kappa}.\label{S_INT_Lambert}
\end{equation}

Finally, I can describe the solution to (\ref{CaseC_relaxedOP}): $\tau$ is the stopping
time such that stopping occurs either at the moment of completion
of the second stage of the project or at $S_{0}^{A}$, conditional on
the absence of the completion of the first stage of the project, i.e.
\begin{equation}
\tau=\begin{cases}
S_{0}^{A}, & \text{if }x_{S_{0}^{A}}=0\\
\tau_{2}\land T, & \text{otherwise},
\end{cases}\label{tau_iii_relaxed_prob}
\end{equation}
where $S_{0}^{A}$ is given by (\ref{S_INT_Lambert}).

\emph{I proceed with the second part of the proof:} I demonstrate that (\ref{tau_iii_relaxed_prob}) satisfies the full system of constraints in (\ref{CaseC_fullOP}), and thus solves (\ref{CaseC_fullOP}). To do this, I need to demonstrate that  $V_{t}\left(\tau\right)\geq0,$
for all $t\in[0,\tau)$.
If the recommendation mechanism $\tau$ is given
by (\ref{tau_iii_relaxed_prob}), then, for $t<S_{0}^{A}$ the absence
of stopping at some $t$ reveals that $x_{t}\neq2$. Thus, 
\[
q_{1}\left(t\right)=\frac{p_{1}\left(t\right)}{p_{1}\left(t\right)+p_{0}\left(t\right)}=\frac{\lambda t}{1+\lambda t},\,\forall t<S_{0}^{A}.
\]
Hence, $\dot{q}_{1}\left(t\right)>0,\text{ for all }t<S_{0}^{A}$. Further,
for $t\geq S_{0}^{A}$, the absence of stopping reveals that $x_{t}=1$.
Thus, $q_{1}\left(t\right)=1,$ for all $t\geq S_{0}^{A}.$

Writing out $V_{t}\left(\tau\right)$ based on (\ref{app_V_dot})
yields 
\begin{equation}
\dot{V_{t}}\left(\tau\right)=\lambda q_{1}\left(t\right)V_{t}\left(\tau\right)+v\lambda\left(\kappa-q_{1}\left(t\right)\right).\label{V_dot_rec_proof}
\end{equation}
$q_{1}\left(0\right)=0$ and $\dot{q}_{1}\left(t\right)>0,\text{ for all }t<S_{0}^{A}$.
I define $\tilde{t}$ as the solution of $\frac{\lambda t}{1+\lambda t}=\kappa$.
$q_{1}\left(t\right)<\kappa$, for all $t\in\left[0,\tilde{t}\land S_{0}^{A}\right]$.

I argue that $V\left(\tau\right)\geq0\Rightarrow V_{t}\left(\tau\right)\geq0,$
for all $t\in\left(0,\tilde{t}\land S_{0}^{A}\right)$. Assume that
this is not true, then $\exists\hat{t}$ such that $\hat{t}\coloneqq\inf\left\{ t\in\left(0,\tilde{t}\land S_{0}^{A}\right):V_{t}(\tau)<0\right\} $.
As $V_{t}(\tau)$ is continuous in $t$, it follows that $V_{\hat{t}}(\tau)=0$,
and by the mean value theorem there must be $\overline{t}\in\left(0,\hat{t}\right)$
such that $\dot{V_{\overline{t}}}\left(\tau\right)\leq0$. But this
is in contradiction with the fact that $V_{\overline{t}}(\tau)\geq0$
and \ref{V_dot_rec_proof}.

Consider now $t\in[\tilde{t}\land S_{0}^{A},\tau)$. The continuation
value can be written as

\begin{equation}
\begin{array}{c}
V_{t}\left(\tau\right)=\Esp\left[\int_{t}^{\tau}v\lambda\left(q_{1}\left(z\right)-\kappa\right)dz\,|\,t<\tau\right].\end{array}\label{ic_iii_V_integr2}
\end{equation}
As $\kappa<\frac{1}{2}$ and $q_{1}\left(t\right)=1,$ for all $t\in[S_{0}^{A},\tau)$,
it holds that $q_{1}\left(t\right)\geq\kappa$, $\forall t\in[\tilde{t}\land S_{0}^{A},\tau)$.
Thus, it can be seen from (\ref{ic_iii_V_integr2}) that $V_{t}\left(\tau\right)\geq0$,
$\forall t\in[\tilde{t}\land S_{0}^{A},\tau)$.


\end{proof}

\begin{proof}[Proof of Proposition \ref{prop:prof_share1}]

I assume it is not the case that $\alpha=1$ and $\beta=0$ as, otherwise, agent is indifferent and discloses no information. I start with proving existence of $\tilde{\kappa}$ and then proceed
to proving that when the project is promising, an investment schedule,
in which stopping never occurs in state $0$, is optimal. Proving
existence of $\tilde{\kappa}$ follows the steps of the proof of Lemma
\ref{kappa_tilde}. The principal's expected payoff is given by 
\[
V\left(\tau\right)=\alpha\p\left(x_{\tau}=2\right)v\Esp\left[e^{-r\tau}|\tau_{2}\leq\tau\right]-\Esp\left[\int_{0}^{\tau}e^{-rs}ds\right]c.
\]
$\tilde{\kappa}$ solves $V\left(\tau_{2}\right)=0$, or, equivalently
\begin{equation}
\alpha\p\left(x_{\tau_{2}\land T}=2\right)v\Esp\left[e^{-r\cdot\tau_{2}\land T}|\tau_{2}\leq T\right]=\Esp\left[\int_{0}^{\tau_{2}\land T}e^{-rs}ds\right]c,\label{eq:robust_IR}
\end{equation}
where $\p\left(x_{\tau_{2}\land T}=2\right)=p_{2}\left(T\right)$.
Solving (\ref{eq:robust_IR}) for $\kappa$ yields 
\[
\tilde{\kappa}\left(T,\lambda,r,\alpha\right)=\frac{1}{\lambda\alpha}\frac{\p\left(x_{\tau_{2}\land T}=2\right)\Esp\left[e^{-r\cdot\tau_{2}\land T}|\tau_{2}\leq T\right]}{\Esp\left[\int_{0}^{\tau_{2}\land T}e^{-rs}ds\right]}.
\]
Finally, $V\left(\tau\right)$ decreases in $\kappa$. Thus, if $\kappa<\tilde{\kappa}\left(T,\lambda,r,\alpha\right)$,
then an investment schedule $\tau=\tau_{2}\land T$ satisfies the
principal's individual rationality constraint.

Consider now the agent's expected payoff $W\left(\tau\right)$ given by
\[
W\left(\tau\right)=\left(1-\alpha\right)\p\left(x_{\tau}=2\right)v\Esp\left[e^{-r\tau}|\tau_{2}\leq\tau\right]+\Esp\left[\int_{0}^{\tau}e^{-rs}ds\right]\beta c.
\]

Consider the case $\kappa\leq\tilde{\kappa}\left(T,\lambda,r,\alpha\right)$.
Consider an investment schedule $\tau$ given by
(\ref{mech_with_S}), i.e., such that stopping happens either immediately at the
moment of the second stage completion, or in state $0$ at a possibly
random interim deadline. Further, consider an alternative investment schedule
$\hat{\tau}=\tau_{2}\land T$. Given the two investment schedules,  $\p\left(x_{\hat{\tau}}=2\right)>\p\left(x_{\tau}=2\right)$.
Further, $\Esp\left[e^{-r\hat{\tau}}|\tau_{2}\leq\hat{\tau}\right]=\Esp\left[e^{-r\tau}|\tau_{2}\leq\tau\right]$
and $\Esp\left[\int_{0}^{\hat{\tau}}e^{-rs}ds\right]>\Esp\left[\int_{0}^{\tau}e^{-rs}ds\right]$.
As $W\left(\hat{\tau}\right)>W\left(\tau\right)$ and $\kappa<\tilde{\kappa}\left(T,\lambda,r,\alpha\right)$,
the agent prefers to implement an investment schedule $\hat{\tau}$
rather than $\tau$.

Consider now the case $\kappa>\tilde{\kappa}\left(T,\lambda,r,\alpha\right)$.
The application of the arguments from the proof of Lemma \ref{lemma_effic_sch2}
establishes the result.


\end{proof}

\subsection{Disclosure of project completion with a deterministic delay}

\label{app_delay}

\begin{proposition}\label{Lambda} Assume $\kappa\in(0,\kappa^{NI}\left(T,\lambda\right)]$
and $T>\bar{S}^{NI}$. The optimal mechanism provides no information until
$t=\bar{S}^{NI}$. At each $t\geq \bar{S}^{NI}$, it generates a recommendation
to stop iff the second stage of the project was completed at date
$\pi\left(t\right)$ in the past, where 
\[
\pi\left(t\right)=-\frac{1}{\lambda}\left(1+\frac{1}{\lambda}\mathcal{W}_{-1}(-\frac{1}{\kappa}e^{-1-\lambda t}\lambda t)\right),
\]
where $\mathcal{W}_{-1}(.)$ denotes the $-1$ branch of Lambert $W$
function.

\end{proposition}

The mechanism from Proposition \ref{Lambda} does not recommend stopping
until the second stage of the project is completed, and thus maximizes
the total surplus. The mechanism makes the principal's individual
rationality constraint bind, $V_{\bar{S}^{NI}}\left(\tau\right)=0$. The
absence of a stopping recommendation after $t=\bar{S}^{NI}$ induces posterior
beliefs $q_{1}\left(t\right)=\kappa,\forall t\geq \bar{S}^{NI}$. Note that
the principal's expected instantaneous payoff within $\Delta t$ is
given by 
\[
v\cdot q_{1}\left(t\right)\lambda\Delta t-c\cdot\Delta t=v\lambda\Delta t\left(q_{1}\left(t\right)-\kappa\right).
\]
No information is provided until $\bar{S}^{NI}$ and after $\bar{S}^{NI}$ the mechanism
keeps the principal's expected instantaneous payoff precisely at $0$,
$\forall t\geq \bar{S}^{NI}$. As a result, the principal's continuation
value is kept at $0$ for all $t\in[\bar{S}^{NI},\tau)$.

The delay is given by $t-\pi\left(t\right)$. At the beginning of
the disclosure, $t=\bar{S}^{NI}$, the delay is $\bar{S}^{NI}$. To keep the belief
regarding state $1$ constant, the delay decreases for all $t\in(\bar{S}^{NI},\tau)$.

\begin{proof}[Proof of Proposition \ref{Lambda}]

Posterior beliefs at date $\pi$ induced by the disclosure of the
absence of second stage completion are given by 
\begin{eqnarray*}
q_{0}\left(\pi\right) & = & \frac{p_{0}\left(\pi\right)}{p_{0}\left(\pi\right)+p_{1}\left(\pi\right)},\\
q_{1}\left(\pi\right) & = & \frac{p_{1}\left(\pi\right)}{p_{0}\left(\pi\right)+p_{1}\left(\pi\right)}.
\end{eqnarray*}
As no other evidence is provided during $(\pi,t]$, the beliefs evolve
according to

\begin{eqnarray*}
q_{0}\left(s\right) & = & \frac{e^{-\lambda s}}{1+\lambda\pi},\\
q_{1}\left(s\right) & = & \frac{e^{-\lambda s}\lambda\left(s+\pi\right)}{1+\lambda\pi},
\end{eqnarray*}
where $s\geq\pi$.

The belief regarding state $1$ at current date $t$ is given by 
\begin{equation}
q_{1}\left(t\right)=\frac{e^{-\lambda\left(t-\pi\right)}\lambda t}{1+\lambda\pi}.\label{q1_delay}
\end{equation}
The dynamic of the state is the same as in the no-information benchmark
until $t=\bar{S}^{NI}$. Therefore, 
\[
q_{0}\left(\bar{S}^{NI}\right)=p_{0}\left(\bar{S}^{NI}\right)=\frac{\kappa}{\lambda \bar{S}^{NI}}\quad\text{and}\quad q_{1}\left(\bar{S}^{NI}\right)=p_{1}\left(\bar{S}^{NI}\right)=\kappa.
\]
The dynamics for $t\geq \bar{S}^{NI}$ then is $q_{1}\left(t\right)=\kappa$,
$\dot{q}_{1}\left(t\right)=0$. Solving from (\ref{q1_delay}), 
\[
\pi=-\frac{1}{\lambda}\left(1+\frac{1}{\lambda}\mathcal{W}_{-1}(-\frac{1}{\kappa}e^{-1-\lambda t}\lambda t)\right).
\]

The recommendation mechanism $\tau$ is obedient. $\tau\geq\tau_{2}\land T$
implies that the recommendation to stop comes only if the second stage
of the project has already been completed, and thus immediate stopping
is clearly optimal for the principal. The recommendation not to stop
is also obedient. $V_{t}\left(\tau\right)\geq0,\forall t\in[0,\bar{S}^{NI})$
is formally demonstrated in the proof of obedience for Proposition
\ref{delayed_tau2}. I proceed by showing that $V_{t}\left(\tau\right)=0,$
$\forall t\in[\bar{S}^{NI},\tau)$. Writing out $V_{t}\left(\tau\right)$
in the recursive form yields 
\[
\begin{aligned}V_{t}\left(\tau\right) & =\left(v\lambda q_{1}\left(t\right)-c\right)\Delta t+\left(1-\lambda q_{1}\left(t\right)\Delta t\right)V_{t+\Delta t}\left(\tau\right)\\
 & =v\lambda\left(q_{1}\left(t\right)-\kappa\right)\Delta t+\left(1-\lambda q_{1}\left(t\right)\Delta t\right)V_{t+\Delta t}\left(\tau\right).
\end{aligned}
\]
As $q_{1}\left(t\right)=\kappa$, $\forall t\in[\bar{S}^{NI},\tau)$, it
becomes 
\[
V_{t}\left(\tau\right)=\left(1-\lambda q_{1}\left(t\right)\Delta t\right)V_{t+\Delta t}\left(\tau\right),\forall t\in[\bar{S}^{NI},\tau).
\]
Differentiating both sides w.r.t. $\Delta t$ yields 
\[
0=-\lambda q_{1}\left(t\right)V_{t+\Delta t}\left(\tau\right)+\dot{V}_{t+\Delta t}\left(\tau\right).
\]
This differential equation together with the boundary condition $V_{T}(\tau)=0$
has a unique solution $V_{t}(\tau)=0$ for all $t\in\left[\bar{S}^{NI},T\right]$.

\end{proof}

\subsection{The case of no project completion deadline}

\label{section_Tinfty}

Importantly, the presence of a hard project deadline $T$ serves as
one of the necessary and sufficient conditions for the agent to commit
to an interim reporting deadline. Without a hard deadline $T$, the
principal's incentives under full information are different. Recall
from Lemma \ref{lemma_FI} the principal's incentive to continue investing
decreases in the length of absence of the first stage completion.
In the case $T\rightarrow\infty$, the continuation value $V_{t|1}^{FI}$
is constant and given by $v\left(1-\kappa\right)$. As a result, the
principal's incentive to continue investing given the absence of stage
completion does not change over time. Thus, if the principal opts
in, he never chooses to stop investing before the completion of the
second stage occurs. As a result, setting an interim deadline stops
serving as an agent's tool to incentivize the principal's investment.
The agent's information policy in the case of no project deadline
is given in Lemma \ref{T_infty}.

\begin{lemma}\label{T_infty}

Assume that $T\rightarrow\infty$. In that case, if $\kappa<\frac{1}{2}$,
then the agent uses the information policy presented in Proposition
\ref{prop:opt_discl}, Case 2.

\end{lemma}

\begin{proof}[Proof of Lemma \ref{T_infty}]

Under full information and the absence of an exogenous deadline, the
principal assigns value $v_{x}$ to each state $x\in\left\{ 0,1,2\right\} $.
Clearly, $v_{2}=v$ as the principal stops immediately and gets $v$.
In state $1$, at each $t$ the principal gets $v\Delta t$ with probability
$\lambda\Delta t$ and pays $c\Delta t$. As a result, the principal's
continuation value is constant. Assume that $\kappa<1$, as otherwise
$c\geq\lambda v$ and the principal chooses not to invest in state
$1$. As the principal's continuation value in state $1$ does not
change over time, 
\[
0=\lambda\cdot(v_{2}-v_{1})-c,
\]
and so 
\[
v_{1}=v-\frac{c}{\lambda}=v(1-\kappa).
\]
Thus, the principal wants to invest in state $0$ if $c\leq\lambda v_{1}$,
i.e., $\kappa\leq\frac{1}{2}$.

Finally, as the information regarding $\tau_{1}$ is not decision-relevant
for the principal, for $\kappa<\frac{1}{2}$, the agent chooses the
information policy that discloses only the completion of the
second stage of the project and optimally postpones the disclosure
to make the principal's individual rationality constraint bind.

\end{proof} 
\end{document}